\documentclass[onecolumn,draftclsnofoot]{IEEEtran}
\usepackage{indentfirst}
\usepackage[dvips]{graphicx}
\usepackage{amsfonts}
\usepackage{multirow}
\usepackage{booktabs}
\usepackage{amsmath,amsthm,amssymb}
\usepackage{color,changepage}
\usepackage{algorithm}
\usepackage{algorithmic}
\usepackage{bbm}
\usepackage{verbatim}
\usepackage{makecell}
\usepackage{subfigure}
\usepackage{mathrsfs}
\usepackage{arydshln}
\usepackage{cases}
\usepackage{threeparttable}
\usepackage{extarrows}
\usepackage{array}
\usepackage{bm}
\usepackage{pifont}
\usepackage[T1]{fontenc}
\usepackage{threeparttable}
\usepackage{booktabs}
\usepackage{mdwlist}
\usepackage{xr}
\usepackage{diagbox}
\usepackage{hyperref}
\usepackage{cite}
\usepackage{xr}
\usepackage{nicematrix}
\usepackage{hhline}
\allowdisplaybreaks[4]

\definecolor{newcolor}{rgb}{0.5,0,1}

\newcommand{\cmarkk}{\ding{56}}

\newcolumntype{I}{!{\vrule width 2pt}}

\newtheorem{Theorem}{Theorem}

\newtheorem{Definition}{Definition}

\newtheorem{Lemma}{Lemma}

\theoremstyle{remark}
\newtheorem{Remark}{Remark}

\DeclareMathAlphabet{\mathpzc}{OT1}{pzc}{m}{it}

\definecolor{newcolor}{rgb}{0.5,0,1}

\def\BibTeX{{\rm B\kern-.05em{\sc i\kern-.025em b}\kern-.08em
    T\kern-.1667em\lower.7ex\hbox{E}\kern-.125emX}}

\let\emptyset\varnothing

\begin{document}

\title{Fundamental Limits of Information-Theoretic Secure Aggregation with User Dropouts in Federated Learning}
\author{Jinbao Zhu, Kai Wan, and Giuseppe Caire
\thanks{Jinbao Zhu is 
with the Information Coding and Transmission (ICT) Key Laboratory of Sichuan Province,
Southwest Jiaotong University, Chengdu 611756, China (email: jinbaozhu@swjtu.edu.cn).
Kai Wan is with the School of Electronic Information and Communications,
Huazhong University of Science and Technology, Wuhan 430074, China (email: kai\_wan@hust.edu.cn).
Giuseppe Caire is with the Electrical Engineering and Computer Science
Department, Technische Universit\"at Berlin, 10587 Berlin, Germany (e-mail:
caire@tu-berlin.de).
}}
\maketitle

\begin{abstract}
Secure aggregation enables a server to compute the sum of users’ local computation results, such as gradients or model updates in federated learning, without revealing their individual values.
When users may drop out during protocol execution, secure aggregation generally relies on correlated random keys to mask the local results while allowing the server to recover the sum of the surviving users’ local results. In this federated learning setting, a secure aggregation protocol consists of three phases: key distribution, masked-gradient upload, and mask elimination. In the key-distribution phase, users establish correlated random keys to protect their local gradients.
Users then use these keys to mask their gradients before uploading them to the server. Because user dropouts may prevent the masks from canceling completely, the users provide additional information in the mask-elimination phase, enabling the server to remove the residual masks and recover the desired gradient sum.
Existing information-theoretic studies typically design the mask-elimination communication for a prescribed dropout threshold and evaluate its cost under the worst-case combination of dropouts across the masked-gradient upload and mask-elimination phases. Such a characterization does not reveal how the required communication varies with the actual dropout scenario, and their schemes cannot exploit such variation to reduce communication. Moreover, the fundamental randomness and communication costs of establishing the correlated keys remain unexplored.
Motivated by these limitations, we formulate an adaptive secure aggregation framework under the $T$-colluding security constraint that supports all feasible dropout scenarios and adapts its communication strategy to the actual dropout scenario observed during protocol execution. We develop an adaptive secure aggregation scheme and derive converse bounds on the randomness needed to generate correlated keys, the communication needed to distribute them, the communication for uploading masked gradients, and the additional communication required to eliminate residual masks. In particular, the latter two bounds are characterized separately for each dropout set realized during the masked-gradient upload phase.
By comparing the achievable rates with the converse bounds, we show that the proposed scheme achieves the information-theoretically optimal masked-gradient communication rate for every feasible dropout scenario. It also achieves the optimal mask-elimination communication rate over almost all feasible dropout regimes.
In addition, the achieved key rate and key-distribution communication rate are each either information-theoretically optimal or within a multiplicative factor of two of their respective optima.
These results establish the first fundamental characterization of the communication and randomness resources required for secure aggregation with user dropouts.
\end{abstract}

\section{Introduction}
Federated learning (FL) has emerged as a prominent distributed machine learning paradigm that enables multiple users to collaboratively train a global model under the coordination of a central server without directly sharing their raw data \cite{mcmahan2017communication,konevcny2016federated,li2020federated,yang2019federated}. 
In a typical FL system, each user performs local training on its private data and uploads only the resulting gradients or model updates to the server, which aggregates them, typically through a linear operation such as summation or averaging, to update the global model for the next training iteration.
By keeping raw training data on local devices, FL significantly reduces the risk of direct privacy leakage. However, recent studies on gradient inversion attacks have demonstrated that sensitive training data can still be reconstructed from the shared gradients \cite{zhu2019deep,zhao2020idlg,geiping2020inverting,yin2021see}. Consequently, protecting individual gradients from disclosure is essential for preserving user privacy in FL. 

Secure aggregation \cite{bonawitz2017practical,bell2020secure,kairouz2021advances} was introduced to provide such privacy protection by enabling the server to recover the aggregate of users' gradients without revealing any individual gradient or compromising model accuracy, making it a key privacy-preserving primitive in federated learning systems. 
In practical deployments, however, secure aggregation protocols must also tolerate user dropouts, which occur when participating users become unavailable during protocol execution.
Such dropouts may result from network failures, limited computational resources, battery depletion, or temporary disconnections \cite{bonawitz2019towards,dean2013tail,ananthanarayanan2013effective}. 
Since both the number and identities of dropout users are unknown in advance and may vary across protocol executions, such uncertainty substantially complicates the design of efficient secure aggregation protocols that can handle all feasible dropout scenarios while preserving privacy guarantees.

To address this difficulty, secure aggregation schemes supporting user dropouts \cite{bonawitz2017practical,bell2020secure,choi2020communication,liu2022efficient} generally rely on correlated random keys held by the users. These keys serve two complementary purposes: each user’s key provides the randomness needed to mask its individual gradient, while the correlation among users’ keys provides the structure needed to eliminate the resulting masks during aggregation. To establish these correlated keys, users first communicate with one another. Each user then uses the established keys to mask its local gradient and uploads the resulting masked gradient to the server.
When some users drop out, however, the resulting masks may not fully cancel, leaving a residual aggregate mask that prevents the server from directly recovering the desired aggregate.
Consequently, additional information must be communicated to the server to eliminate the residual aggregate mask and complete the aggregation.
These operations naturally divide the communication process into three phases: a key-distribution phase, a masked-gradient upload phase, and a mask-elimination phase. The key-distribution phase is performed as an offline setup before the aggregation begins and is therefore assumed to be completed while all users are available. In contrast, both the masked-gradient upload and mask-elimination phases occur during online aggregation and are affected by user dropouts. Dropouts during the masked-gradient upload phase determine which masked gradients are available for aggregation, whereas dropouts during the mask-elimination phase determine which users remain available to provide the information required to eliminate the residual aggregate mask. 
Characterizing the fundamental communication costs incurred in these phases, together with the randomness for establishing the correlated keys, has become a central problem in the information-theoretic study of secure aggregation with dropouts.

The information-theoretic secure aggregation problem with user dropouts was first investigated by Zhao and Sun \cite{zhao2022information} under the $T$-colluding security constraint, where the server aims to recover the sum of the local gradients of the surviving users while learning no information about the individual gradients of the non-colluding users beyond what is revealed by the desired sum and the gradients of the colluding users. 
Their dropout model assumes that only users that successfully upload their masked gradients can participate in the subsequent mask-elimination phase; equivalently, the dropout set in the masked-gradient upload phase is required to be a subset of that in the mask-elimination phase.
Under this dropout model, Zhao and Sun \cite{zhao2022information} characterized the optimal communication overheads of both online phases. Specifically, the communication overhead of the masked-gradient upload phase was characterized directly, whereas that of the mask-elimination phase was characterized by taking the worst case over the admissible dropout scenarios across the two online phases. They also constructed an achievable scheme, albeit under a key-distribution model that assumes the existence of a trusted third party.
Subsequently, So \textit{et al.} \cite{so2022lightsecagg} proposed a lightweight secure aggregation scheme based on Shamir’s secret sharing \cite{shamir1979share} for correlated random key distribution, thereby removing the need for a trusted third party and reducing the required random key size.
Further work by Wan \textit{et al.} \cite{wan2024information,wan2024capacity,zhang2025securegroup} studied the secure aggregation problem under the symmetric groupwise key-distribution model, in which every subset of users of the same size shares one uncoded common random key and all such shared keys have the same entropy size. 
Under this prescribed key-distribution structure, they characterized the optimal communication overheads of both online phases, with the mask-elimination communication evaluated under the worst-case dropout scenarios across the two phases. Beyond the secure aggregation problem with user dropouts, numerous variants of information-theoretic secure aggregation have been investigated, including secure aggregation without user dropouts \cite{Yi2026Capacity,zhao2023secure}, hierarchical secure aggregation \cite{zhang2025optimal,li2026fundamental}, vector secure aggregation \cite{yuan2025vector,hu2026capacity}, and sparse secure aggregation \cite{Tang2026Sparse,sami2024secure}. 
Collectively, these works have substantially advanced the understanding of information-theoretic secure aggregation under various system models and privacy constraints.

Despite these advances, existing information-theoretic secure aggregation works with user dropouts \cite{zhao2022information,so2022lightsecagg,wan2024information,wan2024capacity,zhang2025securegroup} still suffer from three fundamental limitations.
\begin{itemize}
\item Existing secure aggregation schemes design the mask-elimination procedure for a predetermined dropout threshold. To guarantee successful aggregation under any dropout scenario up to this threshold, their mask-elimination communication is determined by the dropout scenario associated with the prescribed threshold, rather than by the dropout scenario actually encountered during protocol execution. This design leads to a fundamental inefficiency. When the actual number of dropouts is smaller than the prescribed threshold, the communication capability of surviving users is underutilized, resulting in unnecessary communication overhead. On the other hand, if the actual number of dropouts exceeds the predetermined threshold, the
aggregation process fails. Therefore, existing schemes cannot simultaneously guarantee
robustness against all feasible dropout scenarios and achieve communication efficiency tailored to the actual number of dropout users during protocol execution. The benefit of adapting communication to the actual number of dropout users has been demonstrated in gradient coding \cite{tandon2017gradient}. In particular, Cao \textit{et al.} \cite{cao2021adaptive} proposed an adaptive gradient coding scheme that achieves
the optimal communication cost for each realized number of user dropouts under a given computation-load constraint. Although both gradient coding and secure aggregation aim to compute the sum of gradients in the presence of user dropouts, they address different challenges. Gradient coding relies on data allocation and redundant computation across users to recover the gradient sum despite user dropouts. In contrast, secure aggregation allows users to compute their gradients locally without sharing their private training data and relies on correlated random keys to protect individual gradients while enabling the recovery
of the desired gradient sum.
Consequently, adapting the mask-elimination communication to the actual dropout scenario while preserving information-theoretic security requires a different approach.
\item Existing works of secure aggregation with user dropouts have primarily focused on characterizing the communication overheads of the masked-gradient upload and mask-elimination phases, while largely overlooking the resources required to establish correlated random keys. Although the fundamental limits of key randomness and key-distribution communication have been fully characterized for secure aggregation without user dropouts \cite{Yi2026Capacity,zhao2023secure}, these results do not directly extend to the setting with user dropouts. This is because the correlated keys must be established before the identities of the surviving users are known, while ensuring both correctness and information-theoretic security under all feasible dropout scenarios. These additional requirements introduce new challenges in determining the minimum randomness and communication resources needed for key establishment. Consequently, the fundamental limits of key randomness and key-distribution communication for secure aggregation with user dropouts remain largely unexplored.
\item Existing works characterize the communication overhead of the
mask-elimination phase by taking the worst case over the dropout scenarios in both the masked-gradient upload and mask-elimination phases. Such a joint worst-case characterization overlooks the fact that the dropout set in the masked-gradient upload phase is already known when mask elimination phase begins and can therefore be used to adapt the subsequent communication.  
Characterizing the optimal mask-elimination communication overhead conditioned on each realized dropout set in the masked-gradient upload phase remains an open problem.
\end{itemize}


Motivated by these limitations, we revisit the $T$-colluding information-theoretic secure aggregation problem with user dropouts in a federated learning system consisting of $N$ users and a central server. 
We consider a general dropout model in which $U$ users may drop out during the masked-gradient upload phase and $S$ users may drop out during the subsequent mask-elimination phase. Unlike prior dropout formulations \cite{zhao2022information,so2022lightsecagg,wan2024information,wan2024capacity,zhang2025securegroup}, we do not impose a set-inclusion relation between the dropout sets in these two phases, allowing user availability to vary during protocol execution. Such variations in user availability may arise from intermittent connectivity, network delays, heterogeneous computation times, or temporary communication failures. For example, a user that does not participate in the masked-gradient upload phase may become available again during the subsequent mask-elimination phase. 
Under this model, we aim to characterize the fundamental randomness and communication costs of secure aggregation with user dropouts. In particular, we investigate the communication overheads of the masked-gradient upload and mask-elimination phases for each realized dropout set in the masked-gradient upload phase, rather than considering only the joint worst-case dropout scenarios across both phases. The main contributions of this paper are summarized as follows.
\begin{itemize}

\item We formulate an adaptive secure aggregation framework with user dropouts, in which a single scheme accommodates all feasible dropout scenarios and its communication strategy dynamically adapts to the actual dropout scenario encountered during protocol execution. We introduce four normalized performance measures: the key rate, key-distribution communication rate, masked-gradient communication rate, and mask-elimination communication rate. For each realized dropout set in the masked-gradient upload phase, we derive converse bounds on the masked-gradient communication rate and the mask-elimination communication rate. Moreover, we establish information-theoretic lower bounds on the key rate and key-distribution communication rate.
\item We propose an adaptive secure aggregation scheme that simultaneously achieves a key rate of $N(T+1)$, a key-distribution communication rate of $N(N-1)$, a masked-gradient communication rate of $N-U$, and a mask-elimination communication rate of $\frac{N-S}{N-S-T}$ for all feasible dropout scenarios satisfying $0\le U\le N-1$ and $0\le S\le N-T-1$. 

\item By comparing the achievable rates with the converse bounds, we show that, for every dropout set $\mathcal{U}$ realized during the masked-gradient upload phase, the proposed scheme achieves the information-theoretically optimal masked-gradient communication rate. For each such dropout set $\mathcal{U}$, it also achieves the optimal mask-elimination communication rate whenever $1\leq S\leq N-T-1$ or $(S=0 \text{ and } T<N-U)$. In these regimes, the optimal mask-elimination communication rate is independent of the dropout set realized during the masked-gradient upload phase. The optimal mask-elimination communication rate remains unknown in the remaining regime $S=0$ and $T\geq N-U$. Finally, the achieved key rate and key-distribution communication rate are each shown to be either optimal or within a multiplicative factor of two of their respective information-theoretic optima.
\end{itemize}

The remainder of this paper is organized as follows.
Section \ref{Problem:formulation} formulates the adaptive secure aggregation problem with user dropouts. Section~\ref{Main:Result} presents the main results, including the achievable rate region, the converse bounds, and the corresponding gap characterization. Section~\ref{Achievability:proof} develops an adaptive secure aggregation scheme. Section \ref{Converse:Proof} establishes the converse bounds. Section~\ref{section:Comparison} provides a detailed comparison with existing information-theoretic secure aggregation works. Finally, this paper is concluded in Section \ref{conclusion}.

\subsubsection*{Notation}
Let boldface capital letters denote matrices and calligraphic letters denote sets. The cardinality of a set $\mathcal{U}$ is denoted by $|\mathcal{U}|$.
For any non-negative integers $m$ and $n$ satisfying $m\leq n$, let $[n]$ and $[m:n]$ denote the sets $\{1,2,\ldots,n\}$ and $\{m,m+1,\ldots,n\}$, respectively.

\section{Problem Formulation}\label{Problem:formulation}
Consider a federated learning system consisting of a central server and $N$ users, where each user $i \in [N]$ holds a private local gradient $W_i$, represented as a vector of $L$ symbols drawn independently and uniformly from a finite field $\mathbb{F}_q$ of size $q$, i.e.,\footnote{It is worth noting that the uniformity and independence assumptions on the gradient vectors play a key role in the converse proof, but are not essential for establishing achievability.}
\begin{IEEEeqnarray}{rCl}
H(W_1,\ldots,W_N)=\sum_{i \in [N]} H(W_i), \label{model:file:inden} \\
H(W_1)=\ldots=H(W_N)=L. \label{infor:indenpe}
\end{IEEEeqnarray}    
In the federated learning system, we assume the existence of error-free private channels among users and between each user and the server to facilitate information-theoretic privacy guarantees. The objective is to recover the sum of the users’ local gradients at the server in the presence of user dropouts, while guaranteeing information-theoretic privacy of each individual gradient against any collusion between the server and up to $T$ users. We consider a passive (honest-but-curious) adversary model, under which the server and colluding users follow the prescribed protocol but may combine all information available to them in an attempt to infer the non-colluding users' gradients. 

To this end, a secure aggregation scheme is divided into three phases: \emph{the key-distribution phase}, \emph{the masked-gradient upload phase}, and \emph{the mask-elimination phase}. The details are presented below.

    \textbf{Key-Distribution Phase:} Each user $i\in[N]$ locally generates a private random key $\mathcal{Z}_i$ over $\mathbb{F}_q$ and then shares an encoded version $\mathcal{Z}_{i,j}$ of this key with every other user $j\in[N]\backslash\{i\}$, i.e.,  
    \begin{equation}\label{model:key:distribution}
        H(\mathcal{Z}_{i,j} |\mathcal{Z}_i) = 0,\quad \forall\, i\in[N], j\in [N]\backslash\{i\}.
    \end{equation}
Since the users’ local random keys are generated independently of their gradient vectors, 
the resulting independence relationships can be stated as follows:
\begin{IEEEeqnarray}{c}\label{model:key:inden}
H(W_1,\ldots,W_N,\mathcal{Z}_1,\ldots,\mathcal{Z}_N)=\sum_{i \in [N]} H(W_i)+\sum_{i \in [N]} H(\mathcal{Z}_i).
\end{IEEEeqnarray}

\textbf{Masked-Gradient Upload Phase:} Upon receiving the encoded keys $\{\mathcal{Z}_{j,i}\}_{j\in[N]\backslash\{i\}}$ from all other users, user $i\in[N]$ generates a masked version $X_i$ of its local gradient $W_i$ using its local private key $\mathcal{Z}_i$ together with the received encoded keys, i.e., 
\begin{equation}\label{model:masked:graidient:2}
        H(X_{i} | W_i,\mathcal{Z}_i, \{\mathcal{Z}_{j,i}\}_{j\in[N]\backslash\{i\}}) = 0, \quad \forall\, i \in [N].
    \end{equation}
    User $i$ then uploads the masked gradient $X_i$ to the server. However, due to user dropouts, the server may only wait to receive masked gradients from a subset of the fastest users. 
    Let $\mathcal{U}\subseteq[N]$ with size $0\leq U\leq N-1$ denote the set of dropout users in the masked-gradient upload phase, i.e., the users whose masked gradients are not received by the server by the prescribed upload deadline.
    Although both the identities of the dropouts $\mathcal{U}$ and their total number $U$ are initially unknown, the server learns them upon completion of the masked-gradient upload phase. The server then receives the masked gradients $\{X_i\}_{i\in[N]\backslash\mathcal{U}}$ and aims to recover the sum of the gradients $W_{sum}^{\mathcal{U}}\triangleq\sum_{i\in[N]\backslash\mathcal{U}}W_i$ of all surviving users $[N]\backslash\mathcal{U}$.
                                                
    \textbf{Mask-Elimination Phase:} 
    To recover the desired gradient sum $W_{sum}^{\mathcal{U}}$, the server broadcasts the identities of the dropout users $\mathcal{U}$ to all users. 
    Upon receiving $\mathcal{U}$, each user $i\in[N]$ stops transmitting the masked gradient $X_i$ and instead generates an ordered sequence of $M$ messages $\mathcal{Y}_{i,1}^{\mathcal{U}},\mathcal{Y}_{i,2}^{\mathcal{U}},\ldots,\mathcal{Y}_{i,M}^{\mathcal{U}}$. These messages are transmitted to the server sequentially in their prescribed order until sufficient messages have been received to recover $W_{sum}^{\mathcal{U}}$, at which point the server broadcasts a stop signal to terminate any further  transmission.     
    These messages are deterministic functions of the information currently available to user $i$, i.e., for any $i\in[N]$,
    \begin{IEEEeqnarray}{c}\label{model:masked:gradients}
  H(\mathcal{Y}_{i,1}^{\mathcal{U}},\mathcal{Y}_{i,2}^{\mathcal{U}},\ldots,\mathcal{Y}_{i,M}^{\mathcal{U}}|W_i,\mathcal{Z}_i,\{\mathcal{Z}_{j,i}\}_{j\in[N]\backslash\{i\}},\mathcal{U})=0. \IEEEeqnarraynumspace
    \end{IEEEeqnarray}
    Due to user dropouts in the mask-elimination phase, the server may receive messages from only a subset of fastest users. More precisely, let $\mathcal{S}\subseteq[N]$ with size $0\leq S\leq S_{max}$ denote the set of dropout users in the mask-elimination phase, i.e., the users that fail to return all mask-elimination messages requested by the server before the phase terminates, where $S_{max}$ denotes the maximum number of tolerated dropouts. We impose no inclusion relation between $\mathcal{U}$ and $\mathcal{S}$. In particular, a user may drop out in one phase but participate in the other. 
    For a realized dropout set $\mathcal{S}$, the server collects the first $M_S$ messages from each surviving user in $[N]\backslash\mathcal{S}$. Here, $M_S$ is a predetermined parameter specifying the number of messages required from each surviving user for successful decoding under $S$ dropouts in the mask-elimination phase. Since $S$ is unknown in advance, the server does not know beforehand which $M_S$ will apply and therefore collects the ordered messages sequentially until sufficient messages have been received for decoding. To ensure communication efficiency, it is natural to require $M_0\leq M_1\leq\ldots\leq M_{S_{\max}}$, reflecting that fewer surviving users may require each of them to provide more messages for successful decoding.

The secure aggregation scheme must satisfy the following two information-theoretic constraints. 
\begin{itemize}
    \item \textbf{Correctness:} It must be guaranteed that the server can decode the desired sum $W_{sum}^{\mathcal{U}}$ for all feasible dropout scenarios satisfying $\mathcal{U}\subseteq[N]$ with size $0\leq U\leq N-1$ and $\mathcal{S}\subseteq[N]$ with size $0\leq S\leq S_{max}$, i.e.,
    \begin{IEEEeqnarray}{c}
        H(W_{sum}^{\mathcal{U}} |\{X_i\}_{i\in[N]\backslash \mathcal{U}},\{\mathcal{Y}_{i,m}^{\mathcal{U}}\}_{i\in[N]\backslash\mathcal{S},m\in[M_S]}) =0. \IEEEeqnarraynumspace\label{eq:correctness}
    \end{IEEEeqnarray}
This constraint allows the server to recover the sum of the gradients of all surviving users $[N]\backslash\mathcal{U}$ in the masked-gradient upload phase, while adapting the mask-elimination communication to the realized dropout scenario in the subsequent phase. In particular, the number of mask-elimination messages required from each surviving user is adjusted according to the realized number of dropouts $S$. Together, these features ensure that the proposed scheme can dynamically adapt to the actual dropout scenario encountered during protocol execution.
    
    \item \textbf{Security:} The collusion between the server and any subset of up to $T$ users must not learn any additional information about the gradients of all users, beyond what is inherently revealed by the desired sum and the local gradients of the $T$ colluding users. Formally, for any subset of colluding users $\mathcal{T}\subseteq[N]$ with $|\mathcal{T}|\leq T$, the following security constraint must hold for all feasible dropout scenarios satisfying $\mathcal{U}\subseteq[N]$ with size $0\leq U\leq N-1$ and $\mathcal{S}\subseteq[N]$ with size $0\leq S\leq S_{max}$:
    \begin{equation}
        I(\{{W}_i\}_{i \in [N]} ; \mathsf{View}_{\mathcal{T}}^{\mathcal{U}} | W_{sum}^{\mathcal{U}}, \{W_i\}_{i\in\mathcal{T}} )=0,
        \label{eq:security}
    \end{equation}
    where $\mathsf{View}_{\mathcal{T}}^{\mathcal{U}}=\left\{\{\mathcal{Z}_i\}_{i\in\mathcal{T}},\{\mathcal{Z}_{i,j},\mathcal{Z}_{j,i}\}_{j\in[N]\backslash\{i\},i\in\mathcal{T}}, \{X_i\}_{i\in[N]}, \{\mathcal{Y}_{i,m}^{\mathcal{U}}\}_{i\in[N],m\in[M]}\right\}$ 
    represents the joint view of the server and the colluding users in $\mathcal{T}$, consisting of all information that may be available to them in the worst case during protocol execution.\footnote{In practical communication systems, a user may be regarded as a dropout because its transmitted data are not received in time, while the same data may still arrive at the server after some delay. To preclude any potential privacy leakage arising from such delayed transmissions, we adopt a stronger security requirement in which the server is assumed to have access to all masked gradients and all mask-elimination messages that may be generated by the users, regardless of whether they are received in time or used for the current aggregation.}   
    \end{itemize}


The performance of the secure aggregation scheme is measured by the following four normalized metrics:\footnote{To simplify the rate definitions, users in the dropout set of each phase are treated as contributing no usable information to the current aggregation, while all surviving users provide the same number of messages.}
\begin{enumerate}
\item[1.] The key rate $R_Z$, defined as the number of key symbols required per desired aggregation symbol to ensure information-theoretic security, is given by
   \begin{IEEEeqnarray}{c}
     R_Z \triangleq\frac{\sum_{i\in[N]}H(\mathcal{Z}_i)}{H(W_{sum}^{\mathcal{U}})}.\label{definition:RZ}
    \end{IEEEeqnarray}
  \item[2.] The key-distribution communication rate $R_K$, defined as the number of  transmitted key symbols per desired aggregation symbol during the key-distribution phase, is given by
    \begin{IEEEeqnarray}{c}  
    R_K\triangleq\frac{\sum_{i\in[N]}\sum_{j\in[N]\backslash\{i\}}H(\mathcal{Z}_{i,j})}{H(W_{sum}^{\mathcal{U}})}. \label{definition:RK}
    \end{IEEEeqnarray}
  \item[3.] The masked-gradient communication rate $R_{\mathcal{U}}$, defined as the number of  transmitted symbols per desired aggregation symbol under the dropout set $\mathcal{U}$ during the masked-gradient upload phase, is given by
    \begin{IEEEeqnarray}{c}  
    R_{\mathcal{U}}=\frac{\sum\limits_{i\in [N]\backslash \mathcal{U}}H(X_i)}{H(W_{sum}^{\mathcal{U}})},\quad \mathcal{U}\subseteq[N], 0\leq |\mathcal{U}|\leq N-1. \label{performance:RU} \label{definition:RKRA}
    \end{IEEEeqnarray}
  \item[4.] The mask-elimination communication rate $R_S^{\mathcal{U}}$ for any given dropout set $\mathcal{U}\subseteq[N]$ with $0\leq|\mathcal{U}|\leq N-1$ in the masked-gradient upload phase, defined as the worst-case number of transmitted symbols per desired aggregation symbol during the mask-elimination phase with $S$ dropout users, is given by
  \begin{IEEEeqnarray}{c}
R_S^{\mathcal{U}}=\frac{\max\limits_{\mathcal{S}\subseteq[N],|\mathcal{S}|=S}\sum\limits_{i\in[N]\backslash\mathcal{S}}\sum\limits_{m\in[M_S]}H(\mathcal{Y}_{i,m}^{\mathcal{U}})}{H(W_{sum}^{\mathcal{U}})},\quad 0\leq S\leq S_{max}. \label{performance:RS}
  \end{IEEEeqnarray}
\end{enumerate}

In Section \ref{App:Infeasibility}, we show that under the correctness and security constraints in \eqref{eq:correctness} and \eqref{eq:security}, the maximum number of dropouts that can be tolerated during the mask-elimination phase is $N-T-1$ for every feasible security parameter $T$, regardless of the dropout set $\mathcal{U}$ during the masked-gradient upload phase. Therefore, throughout this paper, we focus on the nontrivial regime $0\leq S\leq N-T-1$ and set $S_{max}=N-T-1$, so that the secure aggregation scheme is robust against all feasible dropout scenarios.

\begin{Definition}[Achievable Schemes and Rate Tuples for Secure Aggregation with User Dropouts]
\label{def:achi}
A secure aggregation scheme is said to be \emph{achievable} if it satisfies the correctness constraint \eqref{eq:correctness} and the security constraint \eqref{eq:security} for all feasible dropout scenarios specified by $\mathcal{U}\subseteq[N]$ with $0\leq U\leq N-1$ and $\mathcal{S}\subseteq[N]$ with $0\leq S\leq N-T-1$. 
Furthermore, for any given dropout set $\mathcal{U}\subseteq[N]$ of size $0\leq U\leq N-1$ and any given dropout parameter $0\leq S\leq N-T-1$, the rate tuple $(R_Z,R_K,R_{\mathcal{U}},R_S^{\mathcal{U}})$ is said to be \emph{achievable} if there exists an adaptive secure aggregation scheme whose key rate and communication rates do not exceed the prescribed values $R_Z$, $R_K$, $R_{\mathcal{U}}$, and $R_S^{\mathcal{U}}$, respectively, under the dropout scenario in which the dropout set during the masked-gradient upload phase is $\mathcal{U}$ and the number of dropout users during the mask-elimination phase is $S$.
\end{Definition}

Let $R_Z^*,R_K^*,R_{\mathcal{U}}^*$, and $R_S^{\mathcal{U}*}$ denote the optimal key rate, key-distribution communication rate, masked-gradient communication rate, and mask-elimination communication rate, respectively, each defined as the infimum of the corresponding rate over all achievable rate tuples.
The objective of this paper is to characterize the fundamental limits of these optimal rates and to develop adaptive secure aggregation schemes that achieve these limits as closely as possible across all feasible dropout scenarios.




\section{Main Results}\label{Main:Result}

The main results of this paper are summarized in the following three theorems. Theorem~\ref{Theorem:1} establishes an achievable rate region for the secure aggregation problem with user dropouts. Theorem~\ref{Theorem:2} provides information-theoretic lower bounds on the optimal rates. Finally, Theorem~\ref{Theorem:3} characterizes the multiplicative gap between the achievable rates and the information-theoretic optima. 
Throughout this paper, we focus on the nontrivial parameter regime $N>1$ and $0\leq T<N-1$.\footnote{
When $T=N$, the server colludes with all users and therefore has direct access to every local gradient, rendering the security constraint vacuous. 
Moreover, when $T=N-1$, the maximum number of tolerated dropouts in the mask-elimination phase satisfies $S_{\max}=N-T-1=0$, which requires all users to participate in this phase. 
In this case, the server need not rely on the messages transmitted during the masked-gradient upload phase to decode $W_{sum}^{\mathcal{U}}$, since it can recover the desired sum $W_1+W_2+\cdots+W_N$ solely from the messages transmitted by all $N$ users during the mask-elimination phase, as in secure aggregation without user dropouts \cite{Yi2026Capacity,zhao2023secure}.
Consequently, the server colluding with any $T=N-1$ users can recover the remaining user's gradient from the aggregation result and the gradients of the colluding users. Therefore, no individual gradient remains protected, rendering the privacy constraint vacuous.}

\begin{Theorem}\label{Theorem:1}
For the secure aggregation problem with user dropouts and any system parameters $N>1$ and $0\leq T<N-1$, there exists a secure aggregation scheme that achieves the rate tuple
\begin{IEEEeqnarray}{c}
\left(R_Z,R_K,R_{\mathcal{U}},R_S^{\mathcal{U}}\right)=\left(N(T+1), N(N-1), N-U, \frac{N-S}{N-T-S} \right) \notag
\end{IEEEeqnarray}
for every feasible dropout scenario encountered during protocol execution, i.e.,
for every dropout set $\mathcal{U}\subseteq[N]$ of size $0\leq U\leq N-1$ during the masked-gradient upload phase and every feasible  number $0\leq S\leq N-T-1$ of dropouts during the
mask-elimination phase. Furthermore, the following rate region is achievable: 
\begin{IEEEeqnarray}{c}
\mathcal{R}_{\mathcal{U},S}=\left\{(R_Z, R_K, R_{\mathcal{U}}, R_S^{\mathcal{U}})\in\mathbb{R}^{4}:R_Z\geq N(T+1), R_K\geq N(N-1), R_{\mathcal{U}}\geq N-U, R_S^{\mathcal{U}}\geq\frac{N-S}{N-T-S}\right\}. \notag
\end{IEEEeqnarray}
\end{Theorem}
The achievability of Theorem~\ref{Theorem:1} is established in Section~\ref{Achievability:proof} through the construction of an adaptive secure aggregation scheme that achieves the rate tuple $\big(N(T+1), N(N-1), N-U, \frac{N-S}{N-T-S} \big)$ for every feasible dropout scenario satisfying $0\leq U\leq N-1$ and $0\leq S\leq N-T-1$.
We next establish converse bounds on the four rates and then compare them with the achievable rates to characterize the optimality and multiplicative gaps of the secure aggregation scheme established in Theorem~\ref{Theorem:1}. 

\begin{Theorem}\label{Theorem:2}
For the secure aggregation problem with user dropouts and any system parameters $N>1$ and $0\leq T<N-1$, the optimal rate tuple $(R_Z^*,R_K^*,R_{\mathcal{U}}^*,R_S^{\mathcal{U}*})$ satisfies the following lower bounds for every dropout set $\mathcal{U}\subseteq[N]$ of size $0\leq U\leq N-1$ during the masked-gradient upload phase and every feasible  number $0\leq S\leq N-T-1$ of dropouts during the mask-elimination phase:
\begin{IEEEeqnarray}{c}
\left\{
\begin{array}{@{}ll}
R_Z^*\geq\left\{
\begin{array}{@{}ll}
\frac{N(N-1)(N-T)}{2(N-T-1)},&\text{if}\,\,\, T+1< N<2T+1\\
N(T+1),&\text{if}\,\,\, N\geq 2T+1
\end{array}\right. \\
R_K^*\geq\frac{N(N-1)(N-T)}{2(N-T-1)}, \,\,\,\quad \text{for all feasible parameters} \\
(N-1)R_Z^*+(N-2T-1)R_K^*\geq N(N-1)(N-T), \quad \text{if}\,\,\, N\geq 2T+1 \\
R_{\mathcal{U}}^*\geq N-U, \,\,\quad \text{for all feasible parameters} \\
R_S^{\mathcal{U}*}\geq \frac{N-S}{N-S-T},\,\,\quad \text{if}\,\,\, 1\leq S\leq N-T-1\,\, \text{or} \,\, \big(S=0\,\, \text{and} \,\, T<N-U\big)
\end{array}\right.  \notag
\end{IEEEeqnarray} 
\end{Theorem}

The proof of Theorem~\ref{Theorem:2} is provided in Section~\ref{Converse:Proof}. 
Comparing these lower bounds with the achievable rates presented in Theorem \ref{Theorem:1} immediately yields the following characterization of the multiplicative gaps.
\begin{Theorem}\label{Theorem:3}
For the secure aggregation scheme, which achieves the rate tuple $\big(N(T+1), N(N-1), N-U, \frac{N-S}{N-T-S} \big)$ for every dropout set $\mathcal{U}\subseteq[N]$ of size $0\leq U\leq N-1$ during the masked-gradient upload phase and every feasible  number $0\leq S\leq N-T-1$ of dropouts during the
mask-elimination phase, the multiplicative gaps between the four achievable rates and their respective information-theoretic optima are characterized as follows:
\begin{IEEEeqnarray}{c}
\left\{
\begin{array}{@{}ll}
\frac{R_Z}{R_Z^*}<2,\quad &\text{if}\,\,\, T+1< N<2T+1 \\
\frac{R_Z}{R_Z^*}=1,\quad &\text{if}\,\,\, N\geq 2T+1 \\
\frac{R_K}{R_K^*}<2,\quad &\text{if}\,\,\, 0\leq T< N-2\\
\frac{R_K}{R_K^*}=1,\quad &\text{if}\,\,\, T= N-2\\
\frac{R_{\mathcal{U}}}{R_{\mathcal{U}}^*}=1, \quad &\text{for all feasible parameters} \\
\frac{R_S^{\mathcal{U}}}{R_S^{\mathcal{U}*}}=1,\quad &\text{if}\,\,\, 1\leq S\leq N-T-1\,\, \text{or} \,\, \big(S=0\,\, \text{and} \,\, T<N-U\big)
\end{array}\right.  \notag
\end{IEEEeqnarray} 
\end{Theorem}

In summary, for each given dropout set realized during the masked-gradient upload phase, the secure aggregation scheme established in Theorem \ref{Theorem:1} achieves the optimal masked-gradient communication rate for all feasible parameters. It also achieves the optimal mask-elimination communication rate whenever $1\leq S\leq N-T-1$ or
$(S=0 \text{ and } T<N-U)$. 
In these regimes, the optimal mask-elimination communication rate is 
$\frac{N-S}{N-T-S}$, which is independent of both the identity and the number of dropout users during the masked-gradient upload phase. 
When $S=0$ and $T\geq N-U$, the
optimality of the mask-elimination communication rate remains open. Moreover, the key rate is optimal whenever $N\geq 2T+1$ and is within a factor of two of  the information-theoretic optimum when $T+1<N<2T+1$. The key-distribution communication rate is within a factor of two of the optimum for all $0\leq T<N-2$ and is optimal when $T=N-2$. 
More specifically, Fig.~\ref{figure:region} provides a graphical illustration of the achievable and converse regions under two representative parameter settings, namely $(N,T)=(6,3)$ and $(N,T)=(10,3)$, corresponding to the regimes $T+1<N<2T+1$ and $N\geq 2T+1$, respectively. Since the achievable and converse bounds coincide for the masked-gradient communication rate and the mask-elimination communication rate in almost all feasible parameters, the figure focuses exclusively on the tradeoff between the key rate $R_Z$ and the key-distribution communication rate $R_K$. The blue region represents the achievable rate region established in Theorem~\ref{Theorem:1}, while the green region corresponds to rate tuples ruled out by the converse bounds in Theorem~\ref{Theorem:2}. The remaining red region represents the unresolved gap between the achievable and converse regions. As characterized in Theorem~\ref{Theorem:3}, this gap is at most a factor of two.

\begin{figure}[htbp]
	\centering
		\subfigure[$(N,T)=(6,3)$]
		{ 
		\includegraphics[width=0.470\textwidth]{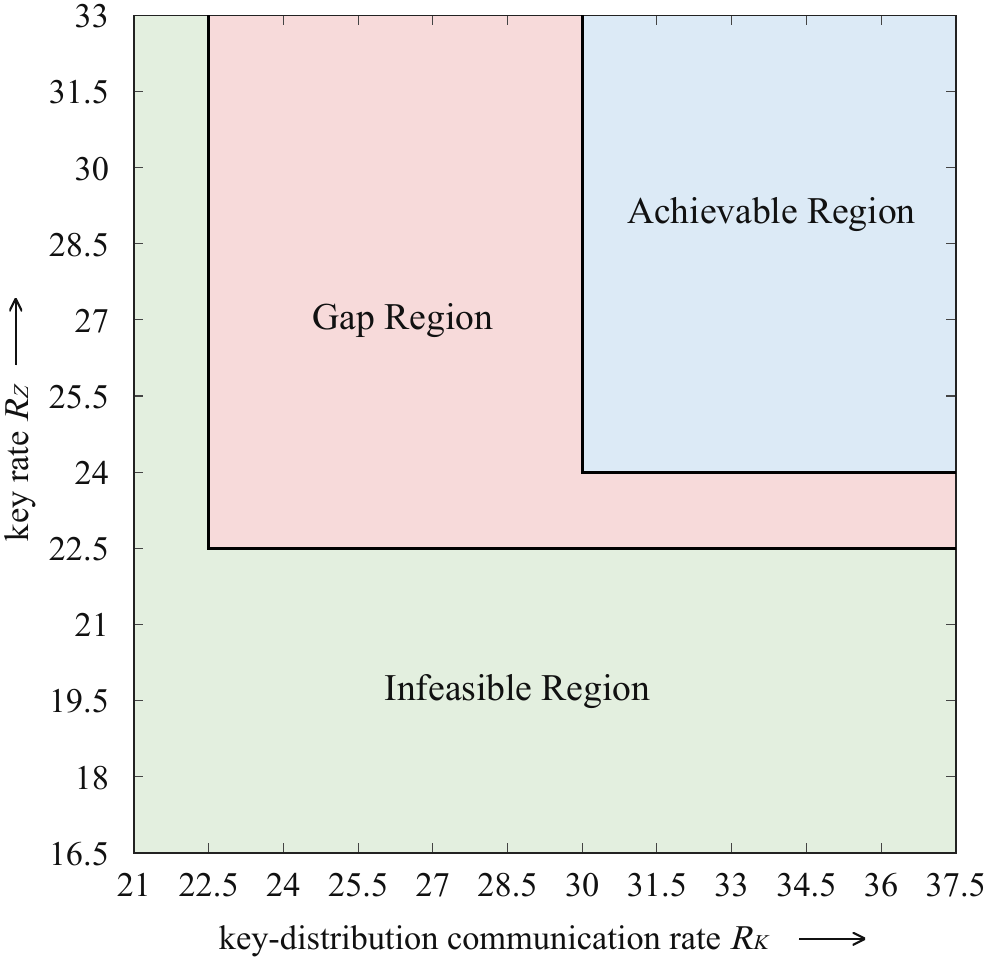}
		}
		\subfigure[$(N,T)=(10,3)$]
		{ 
		\includegraphics[width=0.481\textwidth]{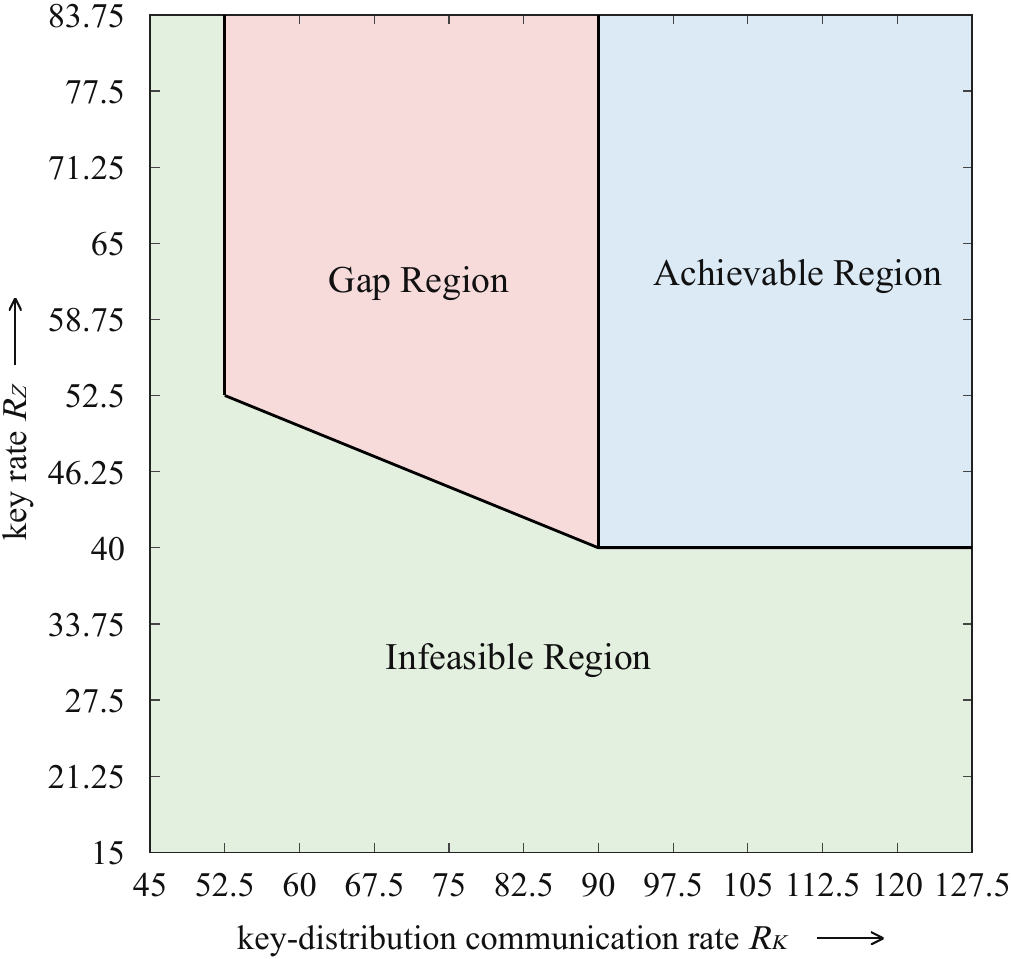}
		}
	\caption{Achievable and converse regions for the key rate $R_Z$ and the key-distribution communication rate $R_K$ under two representative parameter settings: $(N,T)=(6,3)$ and $(N,T)=(10,3)$. The blue region is achievable, whereas the green region is ruled out by the converse bounds. The remaining red region represents a gap of at most a factor of two.}
	\label{figure:region}
\end{figure}

\begin{Remark}
For the $T$-colluding secure aggregation problem without user dropouts (i.e., $U=0$), references~\cite{Yi2026Capacity,zhao2023secure} showed that the desired gradient sum can be recovered by having each user upload only one masked gradient, after which the server sums all received masked gradients to obtain it, without requiring any additional mask-elimination communication. This naturally raises the following question: 
\emph{For the $T$-colluding secure aggregation problem with user dropouts, can one design an adaptive secure aggregation scheme that requires no mask-elimination messages when $U=0$ and invokes the mask-elimination phase only when $U>0$?}
Our converse result answers this question in the negative. Specifically, the information-theoretic lower bounds in Theorem~\ref{Theorem:2} show that $R_{\mathcal{U}}^*\geq N-U$ and $R_S^{\mathcal{U}*}\geq \frac{N-S}{N-T-S}$ for each given dropout set $\mathcal{U}\subseteq[N]$ of size $0\leq U\leq N-1$ during the masked-gradient upload phase. Since $R_S^{\mathcal{U}*}$ remains strictly positive even when $U=0$, the mask-elimination phase is information-theoretically necessary, regardless of whether any dropout actually occurs in the masked-gradient upload phase. 
\end{Remark}


\section{Proof of Theorem~\ref{Theorem:1}: Achievability} \label{Achievability:proof}
In this section, we develop an adaptive secure aggregation scheme that achieves the rate tuple $\big(N(T+1), N(N-1), N-U, \frac{N-S}{N-T-S}\big)$ for all feasible dropout scenarios. The key idea is to establish correlated random keys through a layered polynomial encoding structure that enables the mask-elimination communication to adapt to the actual number of dropouts.
Specifically, each user generates a random mask of length $L$ to protect its local gradient, together with $TL$ additional random symbols to securely distribute the random mask among users while protecting it against up to $T$ colluding users, resulting in a key rate of $N(T+1)$. 
The random mask is then organized into $N-T$ matrices, each representing an encoding layer. In particular, each subsequent layer is constructed by reusing selected symbols from the preceding layers, such that the first $m$ layers collectively enable the recovery of the aggregate mask in the presence of $m-1$ dropouts during the mask-elimination phase, for every $m\in[N-T]$.
The layers are securely encoded into polynomials, whose evaluations are distributed among the users during the key-distribution phase. Each user distributes a total of $L$ encoded symbols to every other user, yielding a key-distribution communication rate of $N(N-1)$.
During the masked-gradient upload phase, each user uploads a masked gradient of length $L$. When $U$ users drop out, the server receives $N-U$ masked gradients, resulting in a communication rate of $N-U$. 
After the server identifies the dropout set $\mathcal{U}$, the users generate an ordered sequence of $N-T$ mask-elimination messages, each constructed from the encoded keys corresponding to one of the $N-T$ encoding layers. Rather than transmitting all these messages according to a prescribed worst-case dropout threshold, each user transmits them sequentially until the server has received sufficient information to recover the aggregate mask. When $S$ users drop out during the mask-elimination phase, the server only needs to collect the first $S+1$ messages from each of the $N-S$ surviving users. The server first recovers the information encoded in the last received layer and then uses the recovered information to successively decode the preceding layers through polynomial interpolation, thereby reconstructing the aggregate mask and recovering the desired gradient sum.
By construction, the first $S+1$ messages contain a total of $\frac{L}{N-T-S}$ symbols per surviving user, yielding a mask-elimination communication rate of $\frac{N-S}{N-T-S}$. Consequently, the established correlated keys can support all feasible dropout scenarios, while the mask-elimination communication dynamically adapts to the actual number of dropouts.

To illustrate the underlying construction, we first present a detailed example.


\subsection{Illustrative Example}
In this example, we consider the secure aggregation problem with system parameters $N=5$ and $T=2$, where the server aims to sum the gradients of all surviving users, while guaranteeing the privacy of each individual gradient against any collusion between the server and up to $T=2$ users.   
To ensure that the proposed scheme can dynamically tolerate all feasible dropouts, we require that $S_{max}=N-T-1=2, M=S_{max}+1=3$, and $L=6$.


\textbf{Key-Distribution Phase:} 
Each user $i\in[5]$ locally generates a random key vector of length $L=6$ over the finite field $\mathbb{F}_q$, given by 
\begin{IEEEeqnarray}{c}
C_i=\left(C_{i,1},C_{i,2},C_{i,3},C_{i,4},C_{i,5},C_{i,6}\right),\quad\forall\,i\in[5].\label{key:example}
\end{IEEEeqnarray}
The key vector $C_i$ can be carefully rearranged into $M=3$ matrices $\mathbf{P}^{(i,1)}$, $\mathbf{P}^{(i,2)}$, and $\mathbf{P}^{(i,3)}$, given by
\begin{IEEEeqnarray}{c}
\mathbf{P}^{(i,1)}=
\left[
\begin{array}{@{\;}cc@{\;}}
C_{i,1} & C_{i,4} \\
C_{i,2} & C_{i,5} \\
C_{i,3} & C_{i,6} \\
\end{array}
\right], \quad
\mathbf{P}^{(i,2)}=
\left[
\begin{array}{@{\;}c@{\;}}
C_{i,3} \\
C_{i,6} \\
\end{array}
\right], \quad
\mathbf{P}^{(i,3)}=
\left[
\begin{array}{@{\;}ccc@{\;}}
C_{i,2} & C_{i,5} & C_{i,6} \\
\end{array}
\right]. \notag
\end{IEEEeqnarray}
Then, for each $m=1,2,3$, the matrix $\mathbf{P}^{(i,m)}$ is secretly shared with all other users to ensure robustness against the presence of $S=m-1$ dropouts during the subsequent mask-elimination phase. More precisely, based on these three matrices, user $i\in[5]$ constructs the following secure encoding polynomials:
\begin{IEEEeqnarray}{rCl}
f_{1}^{(i,1)}(x) & = & D_{1,1}^{(i,1)} + D_{2,1}^{(i,1)}x + C_{i,1}x^2 + C_{i,2}x^3 + C_{i,3}x^4, \label{encoding:poly:1} \IEEEeqnarraynumspace \\
f_{2}^{(i,1)}(x) &  = & D_{1,2}^{(i,1)} + D_{2,2}^{(i,1)}x + C_{i,4}x^2 + C_{i,5}x^3 + C_{i,6}x^4, \label{encoding:poly:2}  \\
f_{1}^{(i,2)}(x) &  = & D_{1,1}^{(i,2)} + D_{2,1}^{(i,2)}x + C_{i,3}x^2 + C_{i,6}x^3, \label{encoding:poly:3}  \\
f_{1}^{(i,3)}(x) &  = & D_{1,1}^{(i,3)} + D_{2,1}^{(i,3)}x + C_{i,2}x^2, \label{encoding:poly:4} \\
f_{2}^{(i,3)}(x) &  = & D_{1,2}^{(i,3)} + D_{2,2}^{(i,3)}x + C_{i,5}x^2,  \label{encoding:poly:5} \\
f_{3}^{(i,3)}(x) &  = & D_{1,3}^{(i,3)} + D_{2,3}^{(i,3)}x + C_{i,6}x^2, \label{encoding:poly:6}
\end{IEEEeqnarray}
where $D_{1,1}^{(i,1)},\ldots, D_{2,3}^{(i,3)}$ are independent random noise symbols locally generated by user $i$, which are used to guarantee the privacy of the key vector $C_i$ against any $T=2$ colluding users. Accordingly, the private random key $\mathcal{Z}_i$ required by  user $i$ is given by
\begin{IEEEeqnarray}{c}
\mathcal{Z}_i=\left\{ C_i, \{D^{(i,1)}_{k,\ell}\}_{k\in[2],\ell\in[2]},\{D^{(i,2)}_{k,1}\}_{k\in[2]},\{D^{(i,3)}_{k,\ell}\}_{k\in[2],\ell\in[3]} \right\}, \quad\forall\, i\in[5]. \notag
\end{IEEEeqnarray}

Let $\alpha_1,\alpha_2,\alpha_3,\alpha_4,\alpha_5$ be distinct elements in $\mathbb{F}_q$.
For any $i,j\in[5]$, the encoded key $\mathcal{Z}_{i,j}$ is obtained by evaluating these secure encoding polynomials at $x=\alpha_j$, i.e.,
\begin{IEEEeqnarray}{c}
\mathcal{Z}_{i,j} = \left\{
f_{1}^{(i,1)}(\alpha_j), f_{2}^{(i,1)}(\alpha_j), f_{1}^{(i,2)}(\alpha_j),
f_{1}^{(i,3)}(\alpha_j), f_{2}^{(i,3)}(\alpha_j), f_{3}^{(i,3)}(\alpha_j)
\right\}. \notag
\end{IEEEeqnarray}
Then, user $i\in [5]$ shares the encoded key $\mathcal{Z}_{i,j}$ with each other user $j\in[5]\backslash\{i\}$.

\textbf{Masked-Gradient Upload Phase:}
Without loss of generality, assume that the local gradient $W_i$ of each user $i\in[5]$ is represented as a row vector. User $i$ locally computes a masked gradient $X_i = W_i + C_i$ using the random key $C_i$ and then uploads it to the server. Suppose that users $3$ and $4$ drop out during this phase, i.e., the dropout set is $\mathcal{U}=\{3,4\}$. Consequently, the server receives only the masked gradients $\{X_1, X_2, X_5\}$ from the surviving users and seeks to recover the desired gradient sum $W_1+W_2+W_5$. To this end, the server computes the sum of the received masked gradients as
\begin{IEEEeqnarray}{c}\notag
X_1+X_2+X_5 = (W_1+W_2+W_5)+(C_1+C_2+C_5). 
\label{eq:aggregation}
\end{IEEEeqnarray}
Therefore, recovering the desired sum $W_1+W_2+W_5$ reduces to eliminating the aggregate mask $C_1+C_2+C_5$.

\textbf{Mask-Elimination Phase:} To eliminate the aggregate mask, the server broadcasts the dropout set $\mathcal{U}=\{3, 4\}$ to all users. Upon receiving $\mathcal{U}$, each user $i\in[5]$ generates and sends the $M=3$ messages $\mathcal{Y}_{i,1}^{\{3,4\}},\mathcal{Y}_{i,2}^{\{3,4\}},\mathcal{Y}_{i,3}^{\{3,4\}}$ in the given order, based on the locally available encoding keys $\{\mathcal{Z}_{j,i}\}_{j \in \{1, 2, 5\}}$, given by
\begin{IEEEeqnarray}{rCl}
\mathcal{Y}_{i,1}^{\{3,4\}} & = & \left\{
    \begin{aligned}
        & f_1^{(1,1)}(\alpha_i) + f_1^{(2,1)}(\alpha_i) + f_1^{(5,1)}(\alpha_i) \\
        & f_2^{(1,1)}(\alpha_i) + f_2^{(2,1)}(\alpha_i) + f_2^{(5,1)}(\alpha_i)
    \end{aligned}
\right\}, \notag \\ 
\mathcal{Y}_{i,2}^{\{3,4\}} & = & \left\{
    f_1^{(1,2)}(\alpha_i) + f_1^{(2,2)}(\alpha_i) + f_1^{(5,2)}(\alpha_i)
\right\}, \notag \\ 
\mathcal{Y}_{i,3}^{\{3,4\}} & = & \left\{
    \begin{aligned}
        & f_1^{(1,3)}(\alpha_i) + f_1^{(2,3)}(\alpha_i) + f_1^{(5,3)}(\alpha_i) \\
        & f_2^{(1,3)}(\alpha_i) + f_2^{(2,3)}(\alpha_i) + f_2^{(5,3)}(\alpha_i) \\
        & f_3^{(1,3)}(\alpha_i) + f_3^{(2,3)}(\alpha_i) + f_3^{(5,3)}(\alpha_i)
    \end{aligned}
\right\}. \notag
\end{IEEEeqnarray}

The proposed secure aggregation scheme can dynamically tolerate the presence of $S=0,1,2$ dropouts during the mask-elimination phase. 
Specifically, when exactly $S$ users drop out for all $S=0,1,2$, the server can recover the desired gradient sum by collecting the first $M_\mathcal{S}=S+1$ messages from each user in any set of $5-S$ surviving users. The corresponding decoding procedures are described below.

\emph{In the presence of $S=0$ dropouts}, the server waits for the first message from each of users $1,2,3,4$, and $5$. It suffices to show that the aggregate mask $C_1+C_2+C_5$ can be decoded from the following received messages:
\begin{IEEEeqnarray}{c}
\left\{\mathcal{Y}_{i,1}^{\{3,4\}}\right\}_{i\in[5]} = \left\{
    \begin{aligned}
        & f_1^{(1,1)}(\alpha_i) + f_1^{(2,1)}(\alpha_i) + f_1^{(5,1)}(\alpha_i) \\
        & f_2^{(1,1)}(\alpha_i) + f_2^{(2,1)}(\alpha_i) + f_2^{(5,1)}(\alpha_i)
    \end{aligned}
\right\}_{i\in[5]}. \notag
\end{IEEEeqnarray}

Since the polynomial $f_{1}^{(1,1)}(x)+f_{1}^{(2,1)}(x)+f_{1}^{(5,1)}(x)$ has degree $4$ by \eqref{encoding:poly:1}, the server can interpolate this polynomial from the received data $\{f_{1}^{(1,1)}(\alpha_i)+f_{1}^{(2,1)}(\alpha_i)+f_{1}^{(5,1)}(\alpha_i)\}_{i\in[5]}$. The coefficients of the interpolated polynomial then yield 
the sums $\{C_{1,1}+C_{2,1}+C_{5,1},C_{1,2}+C_{2,2}+C_{5,2},C_{1,3}+C_{2,3}+C_{5,3}\}$. Similarly, by \eqref{encoding:poly:2}, the server can recover the sums $\{C_{1,4}+C_{2,4}+C_{5,4},C_{1,5}+C_{2,5}+C_{5,5},C_{1,6}+C_{2,6}+C_{5,6}\}$ from the received data $\{f_{2}^{(1,1)}(\alpha_i)+f_{2}^{(2,1)}(\alpha_i)+f_{2}^{(5,1)}(\alpha_i)\}_{i\in[5]}$. Therefore, the aggregate mask $C_1+C_2+C_5$ can be recovered according to \eqref{key:example}.

\emph{In the presence of $S=1$ dropouts}, without loss of generality, assume that user $5$ drops out. The server waits for the first two messages from each of users $1,2,3$, and $4$.
It remains to show that the aggregate mask $C_1+C_2+C_5$ can be recovered from the following received messages:
\begin{IEEEeqnarray}{c}
\left\{
    \begin{aligned}
        \mathcal{Y}_{i,1}^{\{3,4\}} \\
        \mathcal{Y}_{i,2}^{\{3,4\}}
    \end{aligned}
\right\}_{i\in[4]}
= \left\{
    \begin{aligned}
        & f_1^{(1,1)}(\alpha_i) + f_1^{(2,1)}(\alpha_i) + f_1^{(5,1)}(\alpha_i) \\
        & f_2^{(1,1)}(\alpha_i) + f_2^{(2,1)}(\alpha_i) + f_2^{(5,1)}(\alpha_i) \\
        & f_1^{(1,2)}(\alpha_i) + f_1^{(2,2)}(\alpha_i) + f_1^{(5,2)}(\alpha_i)
    \end{aligned}
\right\}_{i\in[4]}. \notag
\end{IEEEeqnarray}
From \eqref{encoding:poly:3}, the polynomial $f_{1}^{(1,2)}(x)+f_{1}^{(2,2)}(x)+f_{1}^{(5,2)}(x)$ of degree $3$ can be interpolated from the received data $\{f_{1}^{(1,2)}(\alpha_i)+f_{1}^{(2,2)}(\alpha_i)+f_{1}^{(5,2)}(\alpha_i)\}_{i\in[4]}$, thereby allowing the server to recover the sums $C_{1,3}+C_{2,3}+C_{5,3}$ and $C_{1,6}+C_{2,6}+C_{5,6}$. Furthermore, the polynomial $f_{1}^{(1,1)}(x)+f_{1}^{(2,1)}(x)+f_{1}^{(5,1)}(x)$ (resp. $f_{2}^{(1,1)}(x)+f_{2}^{(2,1)}(x)+f_{2}^{(5,1)}(x)$) of degree $4$ can be interpolated from the received data $\{f_{1}^{(1,1)}(\alpha_i)+f_{1}^{(2,1)}(\alpha_i)+f_{1}^{(5,1)}(\alpha_i)\}_{i\in[4]}$ (resp. $\{f_{2}^{(1,1)}(\alpha_i)+f_{2}^{(2,1)}(\alpha_i)+f_{2}^{(5,1)}(\alpha_i)\}_{i\in[4]}$) along with the previously recovered sums $C_{1,3}+C_{2,3}+C_{5,3}$ (resp. $C_{1,6}+C_{2,6}+C_{5,6}$). Therefore, the aggregate mask  $C_1+C_2+C_5$ can be recovered by \eqref{encoding:poly:1}-\eqref{encoding:poly:2}.

\emph{In the presence of $S=2$ dropouts}, without loss of generality, assume that users $4$ and $5$ drop out. The server waits for the first three message from each of users $1,2$, and $3$.
It remains to show that the aggregate mask $C_1+C_2+C_5$ can be decoded from the following received messages:
\begin{IEEEeqnarray}{c}
\left\{
    \begin{aligned}
        \mathcal{Y}_{i,1}^{\{3,4\}} \\
        \mathcal{Y}_{i,2}^{\{3,4\}} \\
        \mathcal{Y}_{i,3}^{\{3,4\}}
    \end{aligned}
\right\}_{i\in[3]}
= \left\{
    \begin{aligned}
        & f_1^{(1,1)}(\alpha_i) + f_1^{(2,1)}(\alpha_i) + f_1^{(5,1)}(\alpha_i) \\
        & f_2^{(1,1)}(\alpha_i) + f_2^{(2,1)}(\alpha_i) + f_2^{(5,1)}(\alpha_i) \\
        & f_1^{(1,2)}(\alpha_i) + f_1^{(2,2)}(\alpha_i) + f_1^{(5,2)}(\alpha_i) \\
        & f_1^{(1,3)}(\alpha_i) + f_1^{(2,3)}(\alpha_i) + f_1^{(5,3)}(\alpha_i) \\
        & f_2^{(1,3)}(\alpha_i) + f_2^{(2,3)}(\alpha_i) + f_2^{(5,3)}(\alpha_i) \\
        & f_3^{(1,3)}(\alpha_i) + f_3^{(2,3)}(\alpha_i) + f_3^{(5,3)}(\alpha_i)
    \end{aligned}
\right\}_{i\in[3]}. \notag
\end{IEEEeqnarray}
The decoding procedure is similar to the above cases and is based on Lagrange polynomial interpolation. It proceeds in three steps: 1) By \eqref{encoding:poly:4}--\eqref{encoding:poly:6}, the sums $C_{1,2}+C_{2,2}+C_{5,2}, C_{1,5}+C_{2,5}+C_{5,5}$, and $C_{1,6}+C_{2,6}+C_{5,6}$ are recovered from the received messages. 
2) By \eqref{encoding:poly:3}, the sum $C_{1,3}+C_{2,3}+C_{5,3}$ is recovered using the received messages along with the previously recovered sum $C_{1,6}+C_{2,6}+C_{5,6}$. 
3) By \eqref{encoding:poly:1}--\eqref{encoding:poly:2}, the polynomials $f_{1}^{(1,1)}(x)+f_{1}^{(2,1)}(x)+f_{1}^{(5,1)}(x)$ and $f_{2}^{(1,1)}(x)+f_{2}^{(2,1)}(x)+f_{2}^{(5,1)}(x)$ are interpolated from the received messages along with the currently available sums $\{C_{1,r}+C_{2,r}+C_{5,r}\}_{r\in\{2,3,5,6\}}$, thereby recovering the aggregate mask $C_1+C_2+C_5$.

In general, during the mask-elimination phase, each user successively computes and transmits up to $M=3$ messages to the server. For any realized dropout set $\mathcal{S}\subseteq[5]$ with size $S\in\{0,1,2\}$, the server can recover the desired gradient sum once every surviving user in $[5]\backslash\mathcal{S}$ has successfully uploaded its first $M_S=S+1$ messages. Therefore, in the mask-elimination phase, the proposed scheme dynamically adapts the number of messages transmitted by each surviving user to the realized dropout scenario.

It is straightforward to verify that the proposed secure aggregation scheme achieves a key rate of $15$, a key-distribution communication rate of $20$, a masked-gradient communication rate of $3$ in the presence of $U=2$ dropouts, and a mask-elimination  communication rate of $\frac{5-S}{3-S}$ in the present of $S$ dropouts for all $S=0,1,2$.

\subsection{General Construction}
In this subsection, we present the general construction of the proposed secure aggregation scheme, which can recover the desired gradient sum $W_{sum}^{\mathcal{U}}$ for all feasible dropout scenarios satisfying $\mathcal{U}\subseteq[N]$ with size $0\leq U\leq N-1$ and $\mathcal{S}\subseteq[N]$ with size $0\leq S\leq N-T-1$. The proposed construction builds upon the principles of Shamir's secret sharing \cite{shamir1979share,mceliece1981sharing} and Staircase secret sharing \cite{bitar2017staircase,Bitar2016ISIT}.

To ensure that the server can adaptively tolerate the presence of $S$ dropouts for all feasible $0\leq S\leq N-T-1$ during the mask-elimination phase, we set $M\!=\!S_{max}\!+\!1\!=\!N\!-\!T$ and choose the length $L$ of the gradient vector as
\begin{IEEEeqnarray}{rCl}
L&=&\mathrm{lcm}\{N-T,N-T-1,\ldots,1\}, \notag
\end{IEEEeqnarray}
where $\mathrm{lcm}(\cdot)$ denotes the least common multiple. 
Let $\alpha_{1},\alpha_{2},\ldots,\alpha_{N}$ be $N$ distinct elements in $\mathbb{F}_q$ that are publicly available to all users and the server.

\textbf{Key-Distribution Phase:} In this phase, each user $i\in[N]$ first locally generates a random key vector $C_i$ of length $L$ over $\mathbb{F}_q$, and then rearranges $C_i$ into $M=N-T$ matrices, denoted by $\mathbf{P}^{(i,1)},\mathbf{P}^{(i,2)},\ldots,\mathbf{P}^{(i,M)}$. For any given $m\in[M]$, the matrix $\mathbf{P}^{(i,m)}$ has dimensions $(N-T-m+1)\times\lambda_m$ and is introduced to provide robustness against the presence of $S=m-1$ dropouts during the mask-elimination phase, where 
\begin{IEEEeqnarray}{c}
\lambda_m = 
\begin{cases}
\dfrac{L}{N-T} &\text{if\,\,} m=1 \\
\dfrac{L}{(N-T-m+1)(N-T-m+2)}  &\text{if\,\,} 2 \leq m \leq M
\end{cases}. \IEEEeqnarraynumspace \label{parameter:lambda}
\end{IEEEeqnarray}
More specifically, these matrices are designed in an iterative manner as follows: 
\begin{itemize}
\item In the initial case, the matrix $\mathbf{P}^{(i,1)}$ with dimensions $(N-T)\times\frac{L}{N-T}$ is constructed by distributing the $L$ elements in $C_i$ across all entries of $\mathbf{P}^{(i,1)}$, given by
\begin{IEEEeqnarray}{c}\label{intical:case:P}
P^{(i,1)}_{k,\ell}=C_{i,(\ell-1)(N-T)+k},\quad\forall\,k\in[N-T],\ell\in\Big[\frac{L}{N-T}\Big]. \IEEEeqnarraynumspace
\end{IEEEeqnarray}
Here, $P^{(i,m)}_{k,\ell}$ denotes the entry in the $k$-th row and $\ell$-th column of the matrix $\mathbf{P}^{(i,m)}$ for any $i\in[N], m\in[M], k\in[N-T-m+1]$, and $\ell\in[\lambda_m]$, and $C_{i,\ell}$ denotes the $\ell$-th element of the vector $C_i$ for any $\ell\in[L]$.
\item 
    In the iterative phase, for any given $m\in[2:M]$, the matrix $\mathbf{P}^{(i,m)}$ with dimensions $(N-T-m+1)\times\lambda_m$ is constructed by distributing the elements from the $(N-T-m+2)$-th row of the previously constructed matrices $\mathbf{P}^{(i,1)},\mathbf{P}^{(i,2)},\ldots,\mathbf{P}^{(i,m-1)}$ across all entries of $\mathbf{P}^{(i,m)}$. This is feasible because the total number of elements in the $(N-T-m+2)$-th row of the matrices $\mathbf{P}^{(i,1)},\mathbf{P}^{(i,2)},\ldots,\mathbf{P}^{(i,m-1)}$ exactly matches the number of entries in $\mathbf{P}^{(i,m)}$, i.e., $\sum_{r\in[m-1]}\lambda_r=\frac{L}{N-T-m+2}=(N-T-m+1)\lambda_m$ by \eqref{parameter:lambda}. 
    Therefore, we can set the entry $P^{(i,m)}_{k,\ell}$ of the matrix $\mathbf{P}^{(i,m)}$ as
    \begin{IEEEeqnarray}{c}\label{matrix:P}  
        P^{(i,m)}_{k,\ell}=P^{(i,m')}_{(N-T-m+2),\ell'}, \;\forall\, k\in[N-T-m+1],\ell\in[\lambda_m], \IEEEeqnarraynumspace
    \end{IEEEeqnarray}
    where for any $k\in[N-T-m+1]$ and $\ell\in[\lambda_m]$, $(m',\ell')$ is a unique pair such that $(\ell-1)(N-T-m+1)+k=\sum_{r\in[m'-1]}\lambda_{r}+\ell'$ for some $m'\in[m-1]$ and $\ell'\in[\lambda_{m'}]$. 
\end{itemize}
We complete the construction of  $\mathbf{P}^{(i,1)},\mathbf{P}^{(i,2)},\ldots,\mathbf{P}^{(i,M)}$ by executing the iterative procedure for $m=2,3,\ldots,M$.

Then, for each $m\in[M]$ and $\ell\in[\lambda_m]$, we construct a secure encoding polynomial $f^{(i,m)}_{\ell}(x)$ of degree at most $N-m$ to encode the $N-T-m+1$ entries in the $\ell$-th column of matrix $\mathbf{P}^{(i,m)}$, given by
\begin{IEEEeqnarray}{c}\label{encoding:polynomial}
f^{(i,m)}_{\ell}(x)=\sum\limits_{k\in[T]} D^{(i,m)}_{k,\ell} x^{k-1}+\sum\limits_{k\in[N-T-m+1]} P^{(i,m)}_{k,\ell}x^{T+k-1}, \IEEEeqnarraynumspace
\end{IEEEeqnarray}
where $D^{(i,m)}_{1,\ell},\ldots,D^{(i,m)}_{T,\ell}$ are $T$ independent random symbols generated locally by user $i$ over $\mathbb{F}_q$, which serve as privacy-preserving noise against any set of up to $T$ colluding users. Overall, the collection of all private random symbols generated by user $i$ is given by
\begin{IEEEeqnarray}{c}\label{scheme;random:key}
\mathcal{Z}_i=\left\{ C_i, \{D^{(i,m)}_{k,\ell}\}_{m\in[M],k\in[T],\ell\in[\lambda_m]}  \right\},\quad\forall\,i\in[N]. \IEEEeqnarraynumspace
\end{IEEEeqnarray}
For any given $i,j\in[N]$, the encoded key $\mathcal{Z}_{i,j}$ is constructed by evaluating the secure encoding polynomial $f^{(i,m)}_{\ell}(x)$ at $x=\alpha_j$ for all $m\in[M]$ and $\ell\in[\lambda_m]$, i.e., 
\begin{IEEEeqnarray}{c}\label{encoded:key}
\mathcal{Z}_{i,j}=\left\{f^{(i,m)}_{\ell}(\alpha_j):m\in[M],\ell\in[\lambda_m]\right\}.
\end{IEEEeqnarray}  
User $i\in[N]$ then shares the encoded key $\mathcal{Z}_{i,j}$ with user $j\in[N]\backslash\{i\}$. This completes the key-distribution phase.

\textbf{Masked-Gradient Upload Phase:} Each user $i\in[N]$ computes the masked gradient $X_i$ by masking its local gradient $W_i$ with the random key $C_i$, i.e.,
\begin{IEEEeqnarray}{c}\label{uplaod:X}
X_i = W_i + C_i,\quad\forall\,i\in[N].
\end{IEEEeqnarray}
User $i$ then uploads $X_i$ to the server. After collecting the masked gradients, the server identifies the set of dropout users $\mathcal{U}$ and computes
\begin{IEEEeqnarray}{c}
\sum_{i\in[N]\backslash\mathcal{U}}X_i=\sum_{i\in[N]\backslash\mathcal{U}}W_i +  \sum_{i\in[N]\backslash\mathcal{U}}C_i, \label{aggregate:mask}
\end{IEEEeqnarray}
where the term $\sum_{i\in[N]\backslash\mathcal{U}}C_i$ corresponds to the aggregate mask. Therefore, the server can recover the desired gradient sum $\sum_{i\in[N]\backslash\mathcal{U}}W_i$ by removing the aggregate mask $\sum_{i\in[N]\backslash\mathcal{U}}C_i$.

\textbf{Mask-Elimination Phase:} 
In order to eliminate the aggregate mask $\sum_{i\in[N]\backslash\mathcal{U}}C_i$, the server broadcasts the dropout set $\mathcal{U}$ to all users. Then, for any $i\in[N]$ and $m\in[M]$,  user $i$ can generate the message $\mathcal{Y}_{i,m}^{\mathcal{U}}$ using the locally available encoded keys $\{\mathcal{Z}_{j,i}\}_{j\in[N]\backslash \mathcal{U}}$ in \eqref{encoded:key}, given by
\begin{IEEEeqnarray}{c}\label{encoding:message}
\mathcal{Y}_{i,m}^{\mathcal{U}}=\left\{f^{(\mathcal{U},m)}_{\ell}(\alpha_i):\ell\in[\lambda_m]\right\},   
\end{IEEEeqnarray}
where for any $\ell\in[\lambda_m]$,
\begin{IEEEeqnarray}{c}
f^{(\mathcal{U},m)}_{\ell}(\alpha_i)=\sum\limits_{j\in[N]\backslash\mathcal{U}} f^{(j,m)}_{\ell}(\alpha_i), \notag
\end{IEEEeqnarray}
which corresponds to the evaluation of the following polynomial $f^{(\mathcal{U},m)}_{\ell}(x)$ of degree at most $N-m$ at $x=\alpha_i$.
\begin{IEEEeqnarray}{l}\label{encoding:polynomial:sum}
f^{(\mathcal{U},m)}_{\ell}(x)=\sum\limits_{k\in[T]} \Big(\sum\limits_{j\in[N]\backslash\mathcal{U}}D^{(j,m)}_{k,\ell}\Big) x^{k-1} +\sum\limits_{k\in[N-T-m+1]} \Big(\sum\limits_{j\in[N]\backslash\mathcal{U}}P^{(j,m)}_{k,\ell}\Big)x^{T+k-1}. \IEEEeqnarraynumspace
\end{IEEEeqnarray}
Notably, the messages $\mathcal{Y}_{i,1}^{\mathcal{U}},\mathcal{Y}_{i,2}^{\mathcal{U}},\ldots,\mathcal{Y}_{i,M}^{\mathcal{U}}$ form an ordered transmission sequence. User $i$ transmits these messages to the server sequentially in the prescribed order until it receives a stop signal from the server. 

Let $M_S=S+1$, i.e., in the presence of $S$ dropouts for any given $0\leq S\leq N-T-1$, the server waits to receive the first $S+1$ messages from each of the surviving users in $[N]\backslash\mathcal{S}$ for any dropout set $\mathcal{S}\subseteq[N]$ with size $S$.
Since $S$ is unknown in advance, the stopping point depends on the realized dropout scenario. For any $0\leq S\leq N-T-1$, once the required $M_S=S+1$ messages have been received from all surviving users under the corresponding dropout scenario, the server broadcasts a stop signal to terminate further message transmission.
In total, for any dropout set $\mathcal{S}\subseteq[N]$ with size $0\leq S\leq N-T-1$, the server collects the following messages:
\begin{IEEEeqnarray}{c}
\left\{\mathcal{Y}_{i,m}^{\mathcal{U}}:i\in[N]\backslash\mathcal{S},m\in[S+1]\right\},\quad\forall\,\mathcal{S}\subseteq[N],|\mathcal{S}|=S. \label{upload:mask:message}
\end{IEEEeqnarray}
By \eqref{encoding:message}, this collection can be equivalently expressed as
\begin{IEEEeqnarray}{c}\label{received:messages}
\left\{f^{(\mathcal{U},m)}_{\ell}(\alpha_i):m\in[S+1],\ell\in[\lambda_m],i\in[N]\backslash\mathcal{S}\right\}. \IEEEeqnarraynumspace
\end{IEEEeqnarray}

Next, we prove that for any given $0\leq S\leq N-T-1$ and any dropout set $\mathcal{S}\subseteq[N]$ with size $S$, the server can decode the aggregate mask $\sum_{j\in[N]\backslash\mathcal{U}}C_j$ from the received messages in \eqref{received:messages}.
Similar to the key-distribution phase, the decoding is carried out in an iterative manner below:
\begin{itemize}
\item In the initial case, we demonstrate that the server can decode the sum of the matrices $\{\mathbf{P}^{(j,S+1)}\}_{j\in[N]\backslash\mathcal{U}}$, i.e., $\sum_{j\in[N]\backslash\mathcal{U}}\mathbf{P}^{(j,S+1)}$, from the received data $\{f^{(\mathcal{U},S+1)}_{\ell}(\alpha_i)\}_{\ell\in[\lambda_{S+1}],i\in[N]\backslash\mathcal{S}}$. By \eqref{encoding:polynomial:sum}, the polynomial $f^{(\mathcal{U},S+1)}_{\ell}(x)$ has degree at most $N-S-1$ for every $\ell\in[\lambda_{S+1}]$. Since the server knows the $N-S$ evaluations $\{f^{(\mathcal{U},S+1)}_{\ell}(\alpha_i)\}_{i\in[N]\backslash\mathcal{S}}$, it can uniquely interpolate $f^{(\mathcal{U},S+1)}_{\ell}(x)$ and then recover the coefficients of the interpolated  polynomial, which contain  the entries $\{\sum_{j\in[N]\backslash\mathcal{U}}P^{(j,S+1)}_{k,\ell}\}_{k\in[N-T-S]}$ in the $\ell$-th column of $\sum_{j\in[N]\backslash\mathcal{U}}\mathbf{P}^{(j,S+1)}$. Repeating this procedure for all $\ell\in[\lambda_{S+1}]$ recovers the entire matrix $\sum_{j\in[N]\backslash\mathcal{U}}\mathbf{P}^{(j,S+1)}$.
\item We next consider the iterative step. For any $m\in[S]$, suppose that the matrices
$\{\sum_{j\in[N]\backslash\mathcal{U}}\mathbf{P}^{(j,m+r)}\}_{r\in[S-m+1]}$ 
have already been recovered. We show that the server can then recover $\sum_{j\in[N]\backslash\mathcal{U}}\mathbf{P}^{(j,m)}$ from the received messages $\{f^{(\mathcal{U},m)}_{\ell}(\alpha_i)\}_{\ell\in[\lambda_m],i\in[N]\backslash\mathcal{S}}$.
By the construction of the matrices $\{\mathbf{P}^{(j,m)}\}_{j\in[N],m\in[M]}$ in \eqref{matrix:P}, for every $r\in[S-m+1]$, the matrix $\sum_{j\in[N]\backslash\mathcal{U}}\mathbf{P}^{(j,m+r)}$ contains the entries in the $(N-T-m-r+2)$-th row of the previous matrices $\sum_{j\in[N]\backslash\mathcal{U}}\mathbf{P}^{(j,1)},\sum_{j\in[N]\backslash\mathcal{U}}\mathbf{P}^{(j,2)},\ldots, \sum_{j\in[N]\backslash\mathcal{U}}\mathbf{P}^{(j,m+r-1)}$. Hence, for any $\ell\in[\lambda_m]$, the entries $\{\sum_{j\in[N]\backslash\mathcal{U}}P^{(j,m)}_{N-T-m-r+2,\ell}\}_{r\in[S-m+1]}$ can be obtained from the already recovered sum matrices $\{\sum_{j\in[N]\backslash\mathcal{U}}\mathbf{P}^{(j,m+r)}\}_{r\in[S-m+1]}$.
Moreover, by \eqref{encoding:polynomial:sum}, the polynomial $f^{(\mathcal{U},m)}_{\ell}(x)$ has degree at most $N-m$. Since the server knows the $N-S$ evaluations $\{f^{(\mathcal{U},m)}_{\ell}(\alpha_i)\}_{i\in[N]\backslash\mathcal{S}}$ together with the $S-m+1$ coefficients $\{\sum_{j\in[N]\backslash\mathcal{U}}P^{(j,m)}_{N-T-m-r+2,\ell}\}_{r\in[S-m+1]}$, the polynomial $f^{(\mathcal{U},m)}_{\ell}(x)$ can be uniquely interpolated.
Finally, the coefficients of the polynomial $f^{(\mathcal{U},m)}_{\ell}(x)$ yield the $\ell$-th column of
$\sum_{j\in[N]\backslash\mathcal{U}}\mathbf{P}^{(j,m)}$ and repeating this procedure for all $\ell\in[\lambda_m]$ recovers the entire matrix $\sum_{j\in[N]\backslash\mathcal{U}}\mathbf{P}^{(j,m)}$. 
\end{itemize}
By iteratively applying the above argument for $m=S,S-1,\ldots,1$, the server recovers the matrix $\sum_{j\in[N]\backslash\mathcal{U}}\mathbf{P}^{(j,1)}$ from the received messages in \eqref{received:messages}. Since this matrix is precisely the aggregate mask $\sum_{j\in[N]\backslash\mathcal{U}}C_j$ by \eqref{intical:case:P}, the server can eliminate the aggregate mask and recover the desired gradient sum  $\sum_{j\in[N]\backslash\mathcal{U}}W_j$ by \eqref{aggregate:mask}.

We now characterize the rates achieved by the proposed scheme. By \eqref{scheme;random:key} and \eqref{parameter:lambda}, the total amount of random keys generated by all users is $N(L+T\sum_{m\in[M]}\lambda_m)=N(T+1)L$. Therefore, 
the achieved key rate is $R_Z=N(T+1)$.
For any distinct $i,j\in[N]$, the encoded key $\mathcal{Z}_{i,j}$ contains $\sum_{m\in[M]}\lambda_m=L$ symbols by \eqref{encoded:key} and \eqref{parameter:lambda}. Since each user sends an encoded key to every other user, the achieved key-distribution communication rate is $R_K=N(N-1)$.
In the masked-gradient upload phase, for any given dropout set  $\mathcal{U}\subseteq[N]$ with size $0\leq U\leq N-1$, the server receives $N-U$ masked gradients, each consisting of $L$ symbols. Hence, the achieved masked-gradient communication rate is $R_{\mathcal{U}}=N-U$. In the mask-elimination phase, for each such  $\mathcal{U}$ and any given dropout set $\mathcal{S}\subseteq[N]$ with size $0\leq S\leq N-T-1$, the server collects the first $S+1$ mask-elimination messages from each of the $N-S$ surviving users by \eqref{upload:mask:message}. Since the $m$-th mask-elimination message contains $\lambda_m$ symbols for $m\in[S+1]$ by \eqref{encoding:message}, the achieved mask-elimination communication rate is $\frac{(N-S)\sum_{m\in[S+1]}\lambda_m}{L}=\frac{N-S}{N-T-S}$.
Consequently, for every feasible pair of dropout sets
$(\mathcal{U},\mathcal{S})$ satisfying $0\!\leq\! U\!\leq\! N\!-\!1$ and $0\!\leq\! S\!\leq\! N\!-\!T\!-\!1$, the proposed scheme achieves the rate tuple:
\begin{IEEEeqnarray}{c}
\left(N(T+1), N(N-1), N-U, \frac{N-S}{N-T-S} \right).   \notag
\end{IEEEeqnarray}

The privacy of the proposed scheme is formally established in the Appendix. Here, we provide  intuition by explaining how the three communication phases jointly ensure the privacy constraint. In the key-distribution phase, each user $i\in[N]$ distributes each symbol of its private random key $C_i$ using $T$ additional independent random noise symbols, ensuring that $C_i$ remains information-theoretically secure against any set of up to $T$ colluding users. In the masked-gradient upload phase, the local gradient $W_i$ is protected by the one-time-pad style mask $C_i$, so that the uploaded masked gradient $X_i=W_i+C_i$ reveals no information about $W_i$. In the mask-elimination phase, the transmitted messages are deterministic functions of the desired gradient sum  $W_{sum}^{\mathcal{U}}$ and the information available to the server and any $T$ colluding users from previous phases. Therefore, conditioned on the desired sum $W_{sum}^{\mathcal{U}}$, these mask-elimination messages provide no additional information beyond what is already available from previous phases.
Together, these observations imply that the server colluding with any up to $T$ users cannot infer any information about individual gradients of the remaining users beyond the desired gradient sum.

\section{Proof of Theorem \ref{Theorem:2}: Converse}\label{Converse:Proof}
In this section, we derive converse bounds for the key rate, the key-distribution communication rate, the masked-gradient communication rate, the mask-elimination communication rate, and the maximum tolerable number of dropouts in the mask-elimination phase. 
Note that converse bounds on the masked-gradient communication rate, the mask-elimination communication rate, and the maximum tolerable number of dropouts were previously established in \cite{zhao2022information} under a joint worst case over the dropout scenarios in both the masked-gradient upload and mask-elimination phases. Since these bounds do not directly apply to our dropout-dependent performance characterization, new converse bounds must be
derived for each given dropout set realized during the masked-gradient upload phase.


Before presenting the formal proof, we first establish four useful lemmas, which provide the key technical tools needed for the subsequent converse argument.
To simplify notation, let $\mathcal{D}_i$ denote the collection of encoded keys that user $i\in[N]$ sends to other users and receives from other users during the key-distribution phase, i.e.,
\begin{IEEEeqnarray}{c}\label{def:Di}
\mathcal{D}_i=\left\{\mathcal{Z}_{i,j},\mathcal{Z}_{j,i}:j\in[N]\backslash\{i\}\right\},\quad\forall\,i\in[N].
\end{IEEEeqnarray}

\begin{Lemma}\label{Lemma:converse}
For any given security parameter $0\leq T<N-1$ and any set of colluding users $\mathcal{T}\subseteq[N]$ with $|\mathcal{T}|\leq T$, let $\mathcal{U}\subseteq[N]$ be an arbitrary subset such that $0\leq |\mathcal{U}|\leq N-1$ and $[N]\backslash\mathcal{T}\nsubseteq\mathcal{U}$, i.e., the dropout set $\mathcal{U}$ does not include all the remaining non-colluding users in $[N]\backslash\mathcal{T}$. 
Then, we have 
\begin{IEEEeqnarray}{c}
I(W_{sum}^{\mathcal{U}}; \{X_i\}_{i \in [N] },\{W_i\}_{i\in\mathcal{T}},\{\mathcal{Z}_i\}_{i\in\mathcal{T}}, \{\mathcal{D}_{i}\}_{i\in\mathcal{T}})=0. \notag
\end{IEEEeqnarray}
\end{Lemma}
\begin{proof}
For any colluding users $\mathcal{T}\subseteq[N]$ with $|\mathcal{T}|\leq T$ and $0\!\leq\! T\!<\!N\!-\!1$, let $\mathcal{U}\subseteq[N]$ be an arbitrary subset satisfying $[N]\backslash\mathcal{T}\nsubseteq\mathcal{U}$.
The lemma is proved by considering separately the two cases $0\!\leq\! |\mathcal{U}|\!<\!N\!-\!1$ and $|\mathcal{U}|=N-1$.

We first consider the case $0\leq |\mathcal{U}|<N-1$.
Since the dropout set $\mathcal{U}\subseteq [N]$ satisfies $[N]\backslash\mathcal{T}\nsubseteq\mathcal{U}$, there exists at least one surviving non-colluding user $k \in [N] \backslash(\mathcal{U}\cup\mathcal{T})$. Let $\mathcal{U}' = \mathcal{U}\cup\{k\}$, which satisfies $1\leq|\mathcal{U}'|\leq N-1$. The proof for this case then proceeds as follows.
\begin{IEEEeqnarray}{rCl}
0&\overset{(a)}{=}& I(\{W_{i}\}_{i\in[N]}; \mathsf{View}_{\mathcal{T}}^{\mathcal{U}'} | W_{sum}^{\mathcal{U}'}, \{W_{i}\}_{i\in\mathcal{T}}) \notag \\
& \ge & I(W_{sum}^{\mathcal{U}}; \{X_i\}_{i \in [N] }, \{\mathcal{Z}_i\}_{i\in\mathcal{T}}, \{\mathcal{D}_{i}\}_{i\in\mathcal{T}} | W_{sum}^{\mathcal{U}'}, \{W_{i}\}_{i \in \mathcal{T}}) \IEEEeqnarraynumspace \notag\\
&=&I(W_{sum}^{\mathcal{U}}; \{X_i\}_{i \in [N] }, \{\mathcal{Z}_i\}_{i\in\mathcal{T}}, \{\mathcal{D}_{i}\}_{i\in\mathcal{T}}, W_{sum}^{\mathcal{U}'}, \{W_{i}\}_{i \in \mathcal{T}}) - I(W_{sum}^{\mathcal{U}}; W_{sum}^{\mathcal{U}'}, \{W_{i}\}_{i \in \mathcal{T}}) \notag\\
&\overset{(b)}{=}&I(W_{sum}^{\mathcal{U}}; \{X_i\}_{i \in [N] }, \{\mathcal{Z}_i\}_{i\in\mathcal{T}}, \{\mathcal{D}_{i}\}_{i\in\mathcal{T}}, W_{sum}^{\mathcal{U}'}, \{W_{i}\}_{i \in \mathcal{T}}) \nonumber \\
&\geq& I(W_{sum}^{\mathcal{U}}; \{X_i\}_{i \in [N] },\{\mathcal{Z}_i\}_{i\in\mathcal{T}}, \{\mathcal{D}_{i}\}_{i\in\mathcal{T}}, \{W_i\}_{i\in\mathcal{T}}), \notag 
\end{IEEEeqnarray}
where $(a)$ follows from the security constraint in \eqref{eq:security}, and 
$(b)$ follows from the facts that the gradients $W_1,W_2,\ldots,W_N$ are mutually independent by \eqref{model:file:inden}, and that the index $k$ satisfies $k\in[N]\backslash\mathcal{U},k\notin[N]\backslash\mathcal{U}'$, and $k\notin\mathcal{T}$, which together imply that $I(W_{sum}^{\mathcal{U}}; W_{sum}^{\mathcal{U}'},\{W_{i}\}_{i \in \mathcal{T}})=0$.

Next, we consider the case $|\mathcal{U}|=N-1$. Let $\mathcal{U}'=[N]\backslash\mathcal{U}$, which satisfies $|\mathcal{U}'|=1$. Denote the unique element of $\mathcal{U}'$ by $k$, i.e., $\mathcal{U}'=\{k\}=[N]\backslash\mathcal{U}$. 
Furthermore, since $[N]\backslash\mathcal{T}\nsubseteq\mathcal{U}$, we have $k\notin\mathcal{T}$.
Then, by the security constraint in \eqref{eq:security}, we obtain
\begin{IEEEeqnarray}{rCl}
0&=& I(\{W_{i}\}_{i\in[N]}; \mathsf{View}_{\mathcal{T}}^{\mathcal{U}'} | W_{sum}^{\mathcal{U}'}, \{W_{i}\}_{i\in\mathcal{T}}) \notag \\
& \ge & I\Big(W_{k}; \{X_i\}_{i \in [N] }, \{\mathcal{Z}_i\}_{i\in\mathcal{T}}, \{\mathcal{D}_{i}\}_{i\in\mathcal{T}} \Big| \sum\limits_{i\in[N]\backslash(\{k\}\cup\mathcal{T})}W_{i}, \{W_{i}\}_{i \in \mathcal{T}}\Big) \IEEEeqnarraynumspace \notag\\
&\overset{(a)}{=}& I\Big(W_{k}; \{X_i\}_{i \in [N]\backslash\mathcal{T} }, \{\mathcal{Z}_i\}_{i\in\mathcal{T}}, \{\mathcal{D}_{i}\}_{i\in\mathcal{T}} \Big| \sum\limits_{i\in[N]\backslash(\{k\}\cup\mathcal{T})}W_{i}, \{W_{i}\}_{i \in \mathcal{T}}\Big) \IEEEeqnarraynumspace \notag\\
&\overset{(b)}{=}& I\Big(W_{k}; \{X_i\}_{i \in [N]\backslash\mathcal{T} }, \{\mathcal{Z}_i\}_{i\in\mathcal{T}}, \{\mathcal{D}_{i}\}_{i\in\mathcal{T}} \Big| \sum\limits_{i\in[N]\backslash(\{k\}\cup\mathcal{T})}W_{i} \Big) \IEEEeqnarraynumspace \notag\\
&=& I\Big(W_{k}; \{X_i\}_{i \in [N]\backslash\mathcal{T} }, \{\mathcal{Z}_i\}_{i\in\mathcal{T}}, \{\mathcal{D}_{i}\}_{i\in\mathcal{T}},\sum\limits_{i\in[N]\backslash(\{k\}\cup\mathcal{T})}W_{i} \Big) \IEEEeqnarraynumspace \notag\\
&\geq& I\big(W_{k}; \{X_i\}_{i \in [N]\backslash\mathcal{T} }, \{\mathcal{Z}_i\}_{i\in\mathcal{T}}, \{\mathcal{D}_{i}\}_{i\in\mathcal{T}} \big) \IEEEeqnarraynumspace \notag\\
&\overset{(c)}{=}& I(W_{k}; \{X_i\}_{i \in [N]\backslash\mathcal{T} },\{\mathcal{Z}_i\}_{i\in\mathcal{T}}, \{\mathcal{D}_{i}\}_{i\in\mathcal{T}}, \{W_i\}_{i\in\mathcal{T}}) \notag\\
&\overset{(d)}{=}& I(W_{k}; \{X_i\}_{i \in [N] },\{\mathcal{Z}_i\}_{i\in\mathcal{T}}, \{\mathcal{D}_{i}\}_{i\in\mathcal{T}}, \{W_i\}_{i\in\mathcal{T}}) \notag\\
&\overset{(e)}{=}& I(W_{sum}^{\mathcal{U}}; \{X_i\}_{i \in [N] },\{\mathcal{Z}_i\}_{i\in\mathcal{T}}, \{\mathcal{D}_{i}\}_{i\in\mathcal{T}}, \{W_i\}_{i\in\mathcal{T}}), \notag 
\end{IEEEeqnarray}
where $(a)$ and $(d)$ follow from the fact that $X_i$ is a deterministic function of $\{W_i,\mathcal{Z}_i,\mathcal{D}_i\}$ by \eqref{model:masked:graidient:2} and \eqref{def:Di} for any $i\in\mathcal{T}$; $(b)$ and $(c)$ follow from the facts that $k\notin\mathcal{T}$ and that $\{W_i\}_{i\in\mathcal{T}}$ is independent of $\{W_i\}_{i\in[N]\backslash\mathcal{T}}$ and $\{\mathcal{Z}_i\}_{i\in[N]}$ by \eqref{model:key:inden};
$(e)$ is because $\{k\}=[N]\backslash\mathcal{U}$.
Consequently, we complete the proof of the lemma.
\end{proof}

\begin{Lemma}
For any given security parameter $0\!\leq\! T\!<\!N\!-\!1$ and any subset of colluding users $\mathcal{T}\!\subseteq\![N]$ with $|\mathcal{T}|\!=\!T$, 
\begin{IEEEeqnarray}{c}\label{converse:lemma:1}
H(\{W_i\}_{i\in[N]\backslash\mathcal{T}}|\{\mathcal{D}_i\}_{i\in\mathcal{T}},\{X_i\}_{i\in[N]\backslash\mathcal{T}})=(N-T)L.    
\end{IEEEeqnarray}
\end{Lemma}
\begin{proof}
The proof of the lemma is divided into two cases according to the value of the security parameter $T$: $T=0$ and $1\leq T<N-1$. We first focus on the former case.
In the case of $T=0$, for 
any given dropout set $\mathcal{U}\subseteq[N]$ satisfying $0\leq |\mathcal{U}|\leq N-1$, we have
\begin{IEEEeqnarray}{rCl}
0&\overset{(a)}{=}&I(\{{W}_i\}_{i \in [N]} ; \{X_i\}_{i\in[N]} | W_{sum}^{\mathcal{U}} ) \notag\\
&=&I(\{{W}_i\}_{i \in [N]},W_{sum}^{\mathcal{U}} ; \{X_i\}_{i\in[N]})-I(W_{sum}^{\mathcal{U}} ; \{X_i\}_{i\in[N]}) \notag\\
&\overset{(b)}{=}&I(\{{W}_i\}_{i \in [N]} ; \{X_i\}_{i\in[N]})-I(W_{sum}^{\mathcal{U}} ; \{X_i\}_{i\in[N]}) \notag\\
&\overset{(c)}{=}&I(\{{W}_i\}_{i \in [N]} ; \{X_i\}_{i\in[N]}) \notag\\
&=&H(\{{W}_i\}_{i \in [N]})-H(\{{W}_i\}_{i \in [N]} | \{X_i\}_{i\in[N]}) \notag\\
&\overset{(d)}{=}&NL-H(\{{W}_i\}_{i \in [N]} | \{X_i\}_{i\in[N]}), \notag
\end{IEEEeqnarray}
where $(a)$ follows from the security constraint in \eqref{eq:security};
$(b)$ follows from the facts that $W_{sum}^{\mathcal{U}}$ is a deterministic function of $\{\{{W}_i\}_{i \in [N]},\mathcal{U}\}$, and that the dropout set $\mathcal{U}$ is publicly known and is treated as a fixed subset of users; $(c)$ follows by Lemma \ref{Lemma:converse} such that $I(W_{sum}^{\mathcal{U}} ; \{X_i\}_{i\in[N]})=0$; $(d)$ is due to \eqref{model:file:inden} and \eqref{infor:indenpe}. Therefore, the proof of the lemma is completed for the case $T=0$.

Next, we turn to the case $1\leq T<N-1$.
For any subset $\mathcal{T}\subseteq[N]$ with size $T$, let $\mathcal{U}\subseteq[N]$ be a dropout set satisfying $0\leq |\mathcal{U}|\leq N-1$ and $[N]\backslash\mathcal{U}\subseteq\mathcal{T}$. 
Then, according to the security constraint in \eqref{eq:security}, 
\begin{IEEEeqnarray}{rCl}
0&=&I(\{{W}_i\}_{i \in [N]} ; \{\mathcal{Z}_i\}_{i\in\mathcal{T}},\{\mathcal{D}_{i}\}_{i\in\mathcal{T}}, \{X_i\}_{i\in[N]} | W_{sum}^{\mathcal{U}}, \{W_i\}_{i\in\mathcal{T}} )  \notag\\
&\overset{(a)}{=}&I(\{{W}_i\}_{i \in [N]\backslash\mathcal{T}} ; \{\mathcal{Z}_i\}_{i\in\mathcal{T}},\{\mathcal{D}_{i}\}_{i\in\mathcal{T}}, \{X_i\}_{i\in[N]} | \{W_i\}_{i\in\mathcal{T}} )  \notag\\
&\geq&I(\{{W}_i\}_{i \in [N]\backslash\mathcal{T}} ; \{\mathcal{D}_i\}_{i\in\mathcal{T}}, \{X_i\}_{i\in[N]\backslash\mathcal{T}} | \{W_i\}_{i\in\mathcal{T}} )  \notag\\
&\overset{(b)}{=}&I(\{{W}_i\}_{i \in [N]\backslash\mathcal{T}} ; \{\mathcal{D}_i\}_{i\in\mathcal{T}}, \{X_i\}_{i\in[N]\backslash\mathcal{T}}),  \label{1234} 
\end{IEEEeqnarray}
where $(a)$ is due to $[N]\backslash\mathcal{U}\subseteq\mathcal{T}$ such that $W_{sum}^{\mathcal{U}}$ can be determined by $\{W_i\}_{i\in\mathcal{T}}$, 
and $(b)$ follows from the independence of $\{W_i\}_{i\in\mathcal{T}}$ from $\{W_i\}_{i \in [N]\backslash\mathcal{T}}$ and $\{\mathcal{Z}_i\}_{i\in[N]}$ by \eqref{model:key:inden}, together with the fact that $\{\mathcal{D}_i\}_{i\in\mathcal{T}}$ and $\{X_i\}_{i\in[N]\backslash\mathcal{T}}$ are the functions of $\{W_i\}_{i \in [N]\backslash\mathcal{T}}$ and $\{\mathcal{Z}_i\}_{i\in[N]}$ by \eqref{model:key:distribution}--\ref{model:masked:graidient:2} and \eqref{def:Di}.
Furthermore, by \eqref{1234},
\begin{IEEEeqnarray}{rCl}
&&H(\{{W}_i\}_{i \in [N]\backslash\mathcal{T}} | \{\mathcal{D}_i\}_{i\in\mathcal{T}}, \{X_i\}_{i\in[N]\backslash\mathcal{T}}) \notag\\
&=&H(\{{W}_i\}_{i \in [N]\backslash\mathcal{T}})-I(\{{W}_i\}_{i \in [N]\backslash\mathcal{T}} ; \{\mathcal{D}_i\}_{i\in\mathcal{T}}, \{X_i\}_{i\in[N]\backslash\mathcal{T}}) \notag\\
&=&H(\{{W}_i\}_{i \in [N]\backslash\mathcal{T}}) \notag\\
&=&(N-T)L, \notag
\end{IEEEeqnarray}
which completes the proof of the lemma for the case $1\leq T<N-1$.
\end{proof}

\begin{Lemma}
For any given security parameter $0\le T<N-1$, any subset of colluding users $\mathcal{T} \subseteq [N]$ with size $T$, and any index $k \in [N]\backslash\mathcal{T}$,
\begin{IEEEeqnarray}{c}\label{converse:lemma:2}
H(\{W_i\}_{i\in[N]\backslash\mathcal{T}}|\{\mathcal{D}_i\}_{i\in\mathcal{T}},\mathcal{D}_k,\{X_i\}_{i\in[N]\backslash\mathcal{T}})=0.    
\end{IEEEeqnarray}
\end{Lemma}
\begin{proof}
For any subset of colluding users $\mathcal{T}\subseteq[N]$ with size $0\leq T<N-1$ and any index $k\in[N]\backslash\mathcal{T}$, the lemma is proved by showing that $H(W_i | \mathcal{D}_i,X_i)=0$ for $i=k$ and $H(W_i|\{\mathcal{D}_j\}_{j\in\mathcal{T}},\mathcal{D}_k,X_i)=0$ for $i\in[N]\backslash(\mathcal{T}\cup\{k\})$.
We first focus on the former case.
For any $i\in[N]$, we have
\begin{IEEEeqnarray}{rCl}
0&\leq&I(W_i;\{\mathcal{Z}_j\}_{j\in[N]\backslash\{i\}} | \mathcal{D}_i,X_i ) \label{lemma2:first:case1}\\
&\leq& I(W_i,X_i,\mathcal{Z}_i;\{\mathcal{Z}_j\}_{j\in[N]\backslash\{i\}} |\mathcal{D}_i) \notag\\
&\overset{(a)}{=}& I(W_i,\mathcal{Z}_i;\{\mathcal{Z}_j\}_{j\in[N]\backslash\{i\}} |\mathcal{D}_i) \notag\\
&\leq&I(W_i,\mathcal{Z}_i,\{\mathcal{Z}_{i,j}\}_{j\in[N]\backslash\{i\}};\{\mathcal{Z}_j,\mathcal{Z}_{j,i}\}_{j\in[N]\backslash\{i\}}) \notag\\
&\overset{(b)}{=}&I(W_i,\mathcal{Z}_i;\{\mathcal{Z}_j\}_{j\in[N]\backslash\{i\}}) \notag\\
&\overset{(c)}{=}&0, \notag
\end{IEEEeqnarray}
where $(a)$ follows from the fact that $X_i$ is a deterministic function of $\{W_i,\mathcal{Z}_i,\mathcal{D}_i\}$ by \eqref{model:masked:graidient:2} and \eqref{def:Di}; $(b)$ follows by \eqref{model:key:distribution}; $(c)$ is due to \eqref{model:key:inden}.
Furthermore, for any $i\in[N]$, let $\mathcal{U}=[N]\backslash\{i\}$ and  $\mathcal{S}\subseteq[N]$ satisfying $i\in\mathcal{S}$ and $1\leq |\mathcal{S}|\leq N-T-1$. Then, by \eqref{lemma2:first:case1},
\begin{IEEEeqnarray}{rCl}
H(W_i | \mathcal{D}_i,X_i )
&=&H(W_i | \mathcal{D}_i, X_i )-I(W_i;\{\mathcal{Z}_j\}_{j\in[N]\backslash\{i\}} | \mathcal{D}_i, X_i )\nonumber\\
&=&H(W_i | \mathcal{D}_i,X_i,\{\mathcal{Z}_j\}_{j\in[N]\backslash\{i\}} )\nonumber\\ 
&\overset{(a)}{=}&H(W_i | \mathcal{D}_i, X_i,\{\mathcal{Z}_j\}_{j\in[N]\backslash\{i\}}, \{W_j\}_{j\in[N]\backslash\{i\}},\mathcal{U} )\label{equation:1111}\\ 
&\overset{(b)}{=}&H(W_i | \mathcal{D}_i, X_i,\{\mathcal{Z}_j\}_{j\in[N]\backslash\{i\}},\{W_j\}_{j\in[N]\backslash\{i\}},\{\mathcal{Y}_{j,m}^{\mathcal{U}}\}_{j\in[N]\backslash\mathcal{S},m\in[M]},\mathcal{U}) \label{equation:2222}\\ 
&\stackrel{(c)}{=}&0,\label{equation:3333}
\end{IEEEeqnarray}
where $(a)$ follows from the facts that $\{W_j\}_{j\in[N]\backslash\{i\}}$ is independent of $\{W_i,\mathcal{D}_i,X_i,\{\mathcal{Z}_j\}_{j\in[N]\backslash\{i\}}\}$ by \eqref{model:key:distribution}, \eqref{model:masked:graidient:2}, \eqref{model:key:inden}, and \eqref{def:Di}, and that the dropout set $\mathcal{U}$ is publicly known and is treated as a fixed subset of users; $(b)$ follows by \eqref{model:masked:gradients}; $(c)$ follows by applying $\mathcal{U}\!=\![N]\backslash\{i\}$ to the correctness constraint in \eqref{eq:correctness}.

We next focus on the latter case.
For any $k \in [N]\backslash\mathcal{T}$ and $i\in[N]\backslash(\mathcal{T}\cup\{k\})$,
\begin{IEEEeqnarray}{rCl}
0 &\leq& I(W_i,X_i;\{\mathcal{Z}_j\}_{j\in\mathcal{T}},\mathcal{Z}_k|\{\mathcal{D}_j\}_{j\in\mathcal{T}},\mathcal{D}_k) \label{lemma:condition}\\
&\overset{(a)}{\leq}& I(W_i,\mathcal{Z}_i,\{\mathcal{Z}_{j,i}\}_{j\in[N]\backslash(\mathcal{T}\cup\{i,k\})};\{\mathcal{Z}_j\}_{j\in\mathcal{T}},\mathcal{Z}_k|\{\mathcal{D}_j\}_{j\in\mathcal{T}},\mathcal{D}_k) \notag\\
&\overset{(b)}{\leq}& I(W_i,\mathcal{Z}_i,\{\mathcal{Z}_{j}\}_{j\in[N]\backslash(\mathcal{T}\cup\{i,k\})};\{\mathcal{Z}_j\}_{j\in\mathcal{T}},\mathcal{Z}_k|\{\mathcal{D}_j\}_{j\in\mathcal{T}},\mathcal{D}_k) \notag\\
&\overset{(c)}{\leq}& I(W_i,\mathcal{Z}_i,\{\mathcal{Z}_{j}\}_{j\in[N]\backslash(\mathcal{T}\cup\{i,k\})};\{\mathcal{Z}_j\}_{j\in\mathcal{T}},\mathcal{Z}_k) \notag\\
&\overset{(d)}{=}&0, \notag
\end{IEEEeqnarray}
where $(a)$ follows from the fact that $X_i$ is a function of $\{W_i,\mathcal{Z}_i,\{\mathcal{Z}_{j,i}\}_{j\in[N]\backslash(\mathcal{T}\cup\{i,k\})},\{\mathcal{D}_j\}_{j\in\mathcal{T}},\mathcal{D}_k\}$ by \eqref{model:masked:graidient:2} and \eqref{def:Di}; $(b)$ and $(c)$ follow by \eqref{model:key:distribution}; $(d)$ is due to \eqref{model:key:inden}.
Furthermore, for any $k \in [N]\backslash\mathcal{T}$ and $i\in[N]\backslash(\mathcal{T}\cup\{k\})$, let $\mathcal{U}=[N]\backslash\{i\}$ and let $\mathcal{S}=[N]\backslash(\mathcal{T}\cup\{k\})$ with size $|\mathcal{S}|=N-T-1$. Then, we have
\begin{IEEEeqnarray}{rCl}
H(W_i|\{\mathcal{D}_j\}_{j\in\mathcal{T}},\mathcal{D}_k,X_i)&\overset{(a)}{=}&H(W_i|\{\mathcal{D}_j\}_{j\in\mathcal{T}},\mathcal{D}_k,X_i,\{\mathcal{Z}_j\}_{j\in\mathcal{T}},\mathcal{Z}_k) \notag\\
&\overset{(b)}{=}&H(W_i|\{\mathcal{D}_j\}_{j\in\mathcal{T}},\mathcal{D}_k,X_i,\{\mathcal{Z}_j\}_{j\in\mathcal{T}},\mathcal{Z}_k,\{W_j\}_{j\in\mathcal{T}},W_k,\mathcal{U}) \notag \\
&\overset{(c)}{=}&H(W_i|\{\mathcal{D}_j\}_{j\in\mathcal{T}},\mathcal{D}_k,X_i,\{\mathcal{Z}_j\}_{j\in\mathcal{T}},\mathcal{Z}_k,\{W_j\}_{j\in\mathcal{T}},W_k,\{\mathcal{Y}_{j,m}^{\mathcal{U}}\}_{j\in\mathcal{T}\cup\{k\},m\in[M]}) \notag \\
&\stackrel{(d)}{=}&0,\notag
\end{IEEEeqnarray}
where $(a)$ follows from the fact that $\{\{\mathcal{Z}_j\}_{j\in\mathcal{T}},\mathcal{Z}_k\}$ is independent of $\{W_i,X_i\}$ conditioned on $\{\mathcal{D}_j\}_{j\in\mathcal{T}},\mathcal{D}_k\}$ by \eqref{lemma:condition}; and $(b),(c)$, and $(d)$ follow from arguments similar to those in \eqref{equation:1111}, \eqref{equation:2222}, and \eqref{equation:3333}, respectively. 

Consequently, for any subset of colluding users $\mathcal{T}\subseteq[N]$ with size $0\leq T<N-1$ and any index $k\in[N]\backslash\mathcal{T}$, 
\begin{IEEEeqnarray}{rCl}
&&H(\{W_i\}_{i\in[N]\backslash\mathcal{T}}|\{\mathcal{D}_i\}_{i\in\mathcal{T}},\mathcal{D}_k,\{X_i\}_{i\in[N]\backslash\mathcal{T}}) \notag\\
&=&H(W_k,\{W_i\}_{i\in[N]\backslash(\mathcal{T}\cup\{k\})}|\{\mathcal{D}_i\}_{i\in\mathcal{T}},\mathcal{D}_k,X_k,\{X_i\}_{i\in[N]\backslash(\mathcal{T}\cup\{k\})}) \notag \\
&=&0, \notag
\end{IEEEeqnarray}
which completes the proof of the lemma. 
\end{proof}

\begin{Lemma}
For any given security parameter $0\le T<N-1$, any subset of colluding users $\mathcal{T} \subseteq [N]$ with size $T$, and any index $k \in [N]\backslash\mathcal{T}$,
\begin{IEEEeqnarray}{c}
H(\{\mathcal{Z}_{k,j}\}_{j\in[N]\backslash(\mathcal{T}\cup\{k\})}|\{\mathcal{Z}_{k,j}\}_{j\in\mathcal{T}})
+\sum\limits_{i\in[N]\backslash(\mathcal{T}\cup\{k\})}H(\mathcal{Z}_{i,k}|\{\mathcal{Z}_{i,j}\}_{j\in\mathcal{T}})\geq (N-T)L.\label{lemma3:converse}
\end{IEEEeqnarray}
\end{Lemma}
\begin{proof}
From \eqref{converse:lemma:1} and \eqref{converse:lemma:2},
\begin{IEEEeqnarray}{rCl}
H(\mathcal{D}_k|\{\mathcal{D}_i\}_{i\in\mathcal{T}})&\geq&H(\mathcal{D}_k|\{\mathcal{D}_i\}_{i\in\mathcal{T}},\{X_i\}_{i\in[N]\backslash\mathcal{T}}) \notag\\
&\geq&I(\{W_i\}_{i\in[N]\backslash\mathcal{T}};\mathcal{D}_k|\{\mathcal{D}_i\}_{i\in\mathcal{T}},\{X_i\}_{i\in[N]\backslash\mathcal{T}}) \notag\\
&=&H(\{W_i\}_{i\in[N]\backslash\mathcal{T}}|\{\mathcal{D}_i\}_{i\in\mathcal{T}},\{X_i\}_{i\in[N]\backslash\mathcal{T}})-H(\{W_i\}_{i\in[N]\backslash\mathcal{T}}|\{\mathcal{D}_i\}_{i\in\mathcal{T}},\mathcal{D}_k,\{X_i\}_{i\in[N]\backslash\mathcal{T}}) \notag\\
&=&(N-T)L. \notag
\end{IEEEeqnarray}
Then, according to the facts that $\{\mathcal{Z}_{i,j}\}_{j\in[N]\backslash\{i\}}$ is a deterministic function of $\mathcal{Z}_i$ for any $i\in[N]$ by \eqref{model:key:distribution}, and that $\mathcal{Z}_1,\mathcal{Z}_2,\ldots,\mathcal{Z}_N$ are mutually independent by \eqref{model:key:inden}, we obtain
\begin{IEEEeqnarray}{rCl}
(N-T)L&\leq& H(\mathcal{D}_k|\{\mathcal{D}_i\}_{i\in\mathcal{T}}) \notag\\
&=& H(\{\mathcal{Z}_{k,j},\mathcal{Z}_{j,k}\}_{j\in[N]\backslash\{k\}}|\{\mathcal{Z}_{i,j},\mathcal{Z}_{j,i}\}_{i\in\mathcal{T},j\in[N]\backslash\{i\}}) \notag\\
&=&H(\{\mathcal{Z}_{k,j}\}_{j\in[N]\backslash(\mathcal{T}\cup\{k\})},\{\mathcal{Z}_{j,k}\}_{j\in[N]\backslash(\mathcal{T}\cup\{k\})}|\{\mathcal{Z}_{i,j},\mathcal{Z}_{j,i}\}_{i\in\mathcal{T},j\in[N]\backslash\{i\}}) \notag\\
&=&H(\{\mathcal{Z}_{k,j}\}_{j\in[N]\backslash(\mathcal{T}\cup\{k\})},\{\mathcal{Z}_{j,k}\}_{j\in[N]\backslash(\mathcal{T}\cup\{k\})}|\{\mathcal{Z}_{j,i}\}_{j\in[N]\backslash\mathcal{T},i\in\mathcal{T}}) \notag\\
&=&H(\{\mathcal{Z}_{k,j}\}_{j\in[N]\backslash(\mathcal{T}\cup\{k\})}|\{\mathcal{Z}_{k,i}\}_{i\in\mathcal{T}})
+H(\{\mathcal{Z}_{j,k}\}_{j\in[N]\backslash(\mathcal{T}\cup\{k\})}|\{\mathcal{Z}_{j,i}\}_{j\in[N]\backslash(\mathcal{T}\cup\{k\}),i\in\mathcal{T}}) \notag\\
&=&H(\{\mathcal{Z}_{k,j}\}_{j\in[N]\backslash(\mathcal{T}\cup\{k\})}|\{\mathcal{Z}_{k,i}\}_{i\in\mathcal{T}})
+\sum\limits_{j\in[N]\backslash(\mathcal{T}\cup\{k\})}H(\mathcal{Z}_{j,k}|\{\mathcal{Z}_{j,i}\}_{i\in\mathcal{T}}).\notag
\end{IEEEeqnarray}
This completes the proof of the lemma.
\end{proof}

\subsection{Maximum Dropout Tolerance in the Mask-Elimination Phase} \label{App:Infeasibility}
In this subsection, we prove that for any achievable secure aggregation scheme, the maximum feasible value of $S_{max}$ is $N-T-1$ for every security parameter $0\leq T<N-1$ and every dropout set $\mathcal{U}\subseteq[N]$ with size $0\leq U\leq N-1$ during the masked-gradient upload phase. From Section \ref{Achievability:proof}, we know that there exists an achievable secure aggregation scheme that can tolerate the presence of $S$ dropout users in the mask-elimination phase for all $0\leq S\leq N-T-1$. Therefore, it remains to prove that no achievable secure aggregation scheme can support $S_{max}\geq N-T$. Since Lemma~\ref{Lemma:converse} holds for every achievable secure aggregation scheme and every security parameter $0\leq T<N-1$, the lemma can be invoked to establish the  converse result. We proceed by contradiction.

Assume that $S_{max}\geq N-T$. Then, for any security parameter $0\leq T<N-1$ and any given dropout set $\mathcal{U}\subseteq[N]$ with size $0\leq U\leq N-1$, there must exist 
a colluding user set $\mathcal{T}\subseteq[N]$ satisfying $|\mathcal{T}|=T$ and $[N]\backslash\mathcal{T}\nsubseteq\mathcal{U}$, as well as a dropout set $\mathcal{S}\subseteq[N]$ satisfying $[N]\backslash\mathcal{S}\subseteq\mathcal{T}$ and $|\mathcal{S}|\leq S_{max}$. Accordingly, we obtain
\begin{IEEEeqnarray}{rCl}
L&=& H(W_{sum}^{\mathcal{U}}) \notag \\
&\overset{(a)}{=}&H(W_{sum}^{\mathcal{U}}) -H(W_{sum}^{\mathcal{U}}|\{X_i\}_{i \in [N]\backslash\mathcal{U} }, \{\mathcal{Y}_{i,m}^{\mathcal{U}}\}_{i \in [N]\setminus{\mathcal{S}},m\in[M_S]}) \notag\\
&=&I(W_{sum}^{\mathcal{U}}; \{X_i\}_{i \in [N]\backslash\mathcal{U} }, \{\mathcal{Y}_{i,m}^{\mathcal{U}}\}_{i \in [N]\setminus{\mathcal{S}},m\in[M_S]}) \notag \\
&\overset{(b)}{\leq}&I(W_{sum}^{\mathcal{U}}; \{X_i\}_{i \in [N]\backslash\mathcal{U} }, \{W_i\}_{i\in\mathcal{T}},\{\mathcal{Z}_i\}_{i\in\mathcal{T}},\{\mathcal{D}_{i}\}_{i\in\mathcal{T}},\mathcal{U}) \notag \\
&\overset{(c)}{=}&I(W_{sum}^{\mathcal{U}}; \{X_i\}_{i \in [N]\backslash\mathcal{U} }, \{W_i\}_{i\in\mathcal{T}},\{\mathcal{Z}_i\}_{i\in\mathcal{T}},\{\mathcal{D}_{i}\}_{i\in\mathcal{T}}) \notag \\
&\overset{(d)}{=}&0, \notag
\end{IEEEeqnarray}
where $(a)$ follows from the correctness constraint in \eqref{eq:correctness};
$(b)$ follows by $[N]\backslash\mathcal{S}\subseteq\mathcal{T}$ and the fact that $\{\mathcal{Y}_{i,m}^{\mathcal{U}}\}_{m\in[M_S]}$ is the deterministic function of $\{W_i,\mathcal{Z}_i, \mathcal{D}_{i},\mathcal{U}\}$ for any $i \in [N]\setminus{\mathcal{S}}$ by \eqref{model:masked:gradients} and \eqref{def:Di};
$(c)$ follows from the fact that the dropout set $\mathcal{U}$ is publicly known and is treated as a fixed subset of users;
$(d)$ is due to Lemma \ref{Lemma:converse}.

This leads to a contradiction because $L\ge 1$. Therefore, for any achievable secure aggregation scheme, the maximum number of dropouts that can be tolerated during the mask-elimination phase must satisfy $S_{max}<N-T$, regardless of the dropout set $\mathcal{U}$ in the masked-gradient upload phase. Since $S_{max}$ is an integer, its maximum feasible value is $N-T-1$.

\subsection{Converse for the Key-Distribution Communication Rate}
In this subsection, we establish a lower bound on the optimal key-distribution communication rate $R_K^*$. For any given security parameter $0\leq T<N-1$, we complete the proof by summing the inequality in \eqref{lemma3:converse} over all subsets $\mathcal{T}\subseteq[N]$ with size $|\mathcal{T}|=T$ and all $k\in[N]\backslash\mathcal{T}$, as follows:
\begin{IEEEeqnarray}{rCl}
\binom{N}{T}(N-T)^2L
&\leq& \sum\limits_{\substack{\mathcal{T}\subseteq [N]\\|\mathcal{T}| = T}}\sum\limits_{k\in[N]\backslash\mathcal{T}}H(\{\mathcal{Z}_{k,j}\}_{j\in[N]\backslash(\mathcal{T}\cup\{k\})}|\{\mathcal{Z}_{k,j}\}_{j\in\mathcal{T}}) \notag\\
&&+\sum\limits_{\substack{\mathcal{T}\subseteq [N]\\|\mathcal{T}| = T}}\sum\limits_{k\in[N]\backslash\mathcal{T}}\sum\limits_{i\in[N]\backslash(\mathcal{T}\cup\{k\})}H(\mathcal{Z}_{i,k}|\{\mathcal{Z}_{i,j}\}_{j\in\mathcal{T}}) \notag\\
&\leq& \sum\limits_{\substack{\mathcal{T}\subseteq [N]\\|\mathcal{T}| = T}}\sum\limits_{k\in[N]\backslash\mathcal{T}}\sum\limits_{j\in[N]\backslash(\mathcal{T}\cup\{k\})}H(\mathcal{Z}_{k,j})
+\sum\limits_{\substack{\mathcal{T}\subseteq [N]\\|\mathcal{T}| = T}}\sum\limits_{k\in[N]\backslash\mathcal{T}}\sum\limits_{i\in[N]\backslash(\mathcal{T}\cup\{k\})}H(\mathcal{Z}_{i,k}) \notag \\
&=&2\cdot \sum\limits_{\substack{\mathcal{T}\subseteq [N]\\|\mathcal{T}| = T}}\sum\limits_{k\in[N]\backslash\mathcal{T}}\sum\limits_{j\in[N]\backslash(\mathcal{T}\cup\{k\})}H(\mathcal{Z}_{k,j}) \notag\\
&=&2\binom{N-2}{T} \sum\limits_{i\in[N]}\sum\limits_{j\in[N]\backslash\{i\}}H(\mathcal{Z}_{i,j}). \notag
\end{IEEEeqnarray}
Then, by \eqref{definition:RK}, we obtain the following lower bound on the optimal key-distribution communication rate $R_K^*$:
 \begin{IEEEeqnarray}{c}  
    R_K^{*}=\frac{\sum_{i\in[N]}\sum_{j\in[N]\backslash\{i\}}H(\mathcal{Z}_{i,j})}{H(W_{sum}^{\mathcal{U}})}\geq\frac{\binom{N}{T}(N-T)^2}{2\binom{N-2}{T}}=\frac{N(N-1)(N-T)}{2(N-T-1)}. \notag
\end{IEEEeqnarray}

\subsection{Converse for the Key Rate}

In this subsection, we derive a lower bound on the optimal key rate $R_Z^*$ and a lower bound on a linear combination of $R_Z^*$ and the optimal key-distribution communication rate $R_K^*$.

For any given security parameter $0\leq T<N-1$, summing the inequality in \eqref{lemma3:converse} over all subsets $\mathcal{T}\subseteq[N]$ with size $|\mathcal{T}|=T$ and all $k\in[N]\backslash\mathcal{T}$ yields
\begin{IEEEeqnarray}{rCl}
\binom{N}{T}(N-T)^2L
&\leq& \sum\limits_{\substack{\mathcal{T}\subseteq [N]\\|\mathcal{T}| = T}}\sum\limits_{k\in[N]\backslash\mathcal{T}}H(\{\mathcal{Z}_{k,j}\}_{j\in[N]\backslash(\mathcal{T}\cup\{k\})}|\{\mathcal{Z}_{k,j}\}_{j\in\mathcal{T}}) \notag\\
&&+\sum\limits_{\substack{\mathcal{T}\subseteq [N]\\|\mathcal{T}| = T}}\sum\limits_{k\in[N]\backslash\mathcal{T}}\sum\limits_{i\in[N]\backslash(\mathcal{T}\cup\{k\})}H(\mathcal{Z}_{i,k}|\{\mathcal{Z}_{i,j}\}_{j\in\mathcal{T}}) \notag \\
&=&\sum\limits_{\substack{\mathcal{T}\subseteq [N]\\|\mathcal{T}| = T}}\sum\limits_{k\in[N]\backslash\mathcal{T}}H(\{\mathcal{Z}_{k,j}\}_{j\in[N]\backslash\{k\}})-\sum\limits_{\substack{\mathcal{T}\subseteq [N]\\|\mathcal{T}| = T}}\sum\limits_{k\in[N]\backslash\mathcal{T}}H(\{\mathcal{Z}_{k,j}\}_{j\in\mathcal{T}}) \notag\\
&&\!+\!\sum\limits_{\substack{\mathcal{T}\subseteq [N]\\|\mathcal{T}| = T}}\sum\limits_{k\in[N]\backslash\mathcal{T}}\sum\limits_{i\in[N]\backslash(\mathcal{T}\cup\{k\})}\!\!H(\{\mathcal{Z}_{i,j}\}_{j\in\mathcal{T}\cup\{k\}})\!-\!\sum\limits_{\substack{\mathcal{T}\subseteq [N]\\|\mathcal{T}| = T}}\sum\limits_{k\in[N]\backslash\mathcal{T}}\sum\limits_{i\in[N]\backslash(\mathcal{T}\cup\{k\})}\!\!H(\{\mathcal{Z}_{i,j}\}_{j\in\mathcal{T}}) \notag\\
&=&\sum\limits_{k\in[N]}\sum\limits_{\substack{\mathcal{T}\subseteq [N]\backslash\{k\}\\|\mathcal{T}|=T}}H(\{\mathcal{Z}_{k,j}\}_{j\in[N]\backslash\{k\}})-\sum\limits_{k\in[N]}\sum\limits_{\substack{\mathcal{T}\subseteq [N]\backslash\{k\}\\|\mathcal{T}|=T}}H(\{\mathcal{Z}_{k,j}\}_{j\in\mathcal{T}}) \notag\\
&&+(T+1)\sum\limits_{i\in[N]}\sum\limits_{\substack{\mathcal{T}\subseteq [N]\backslash\{i\}\\|\mathcal{T}|=T+1}}\!\!\!H(\{\mathcal{Z}_{i,j}\}_{j\in\mathcal{T}}) \!-\!\sum\limits_{i\in[N]}\sum\limits_{\substack{\mathcal{T}\subseteq [N]\backslash\{i\}\\|\mathcal{T}|=T}}\sum\limits_{k\in[N]\backslash(\mathcal{T}\cup\{i\})}\!\!\!H(\{\mathcal{Z}_{i,j}\}_{j\in\mathcal{T}}) \notag \\
&\overset{(a)}{\leq}&\binom{N-1}{T}\sum\limits_{k\in[N]}H(\{\mathcal{Z}_{k,j}\}_{j\in[N]\backslash\{k\}})-\frac{\binom{N-1}{T}}{\binom{N-1}{T+1}}\frac{T}{T+1}\sum\limits_{k\in[N]}\sum\limits_{\substack{\mathcal{T}\subseteq [N]\backslash\{k\}\\|\mathcal{T}|=T+1}}\!\!\!H(\{\mathcal{Z}_{k,j}\}_{j\in\mathcal{T}}) \notag\\
&&+(T+1)\sum\limits_{i\in[N]}\sum\limits_{\substack{\mathcal{T}\subseteq [N]\backslash\{i\}\\|\mathcal{T}|=T+1}}\!\!\!H(\{\mathcal{Z}_{i,j}\}_{j\in\mathcal{T}}) \notag\\
&&-(N-T-1)\frac{\binom{N-1}{T}}{\binom{N-1}{T+1}}\frac{T}{T+1}\sum\limits_{i\in[N]}\sum\limits_{\substack{\mathcal{T}\subseteq [N]\backslash\{i\}\\|\mathcal{T}|=T+1}}\!\!\!H(\{\mathcal{Z}_{i,j}\}_{j\in\mathcal{T}}) \notag\\
&=&\binom{N-1}{T}\sum\limits_{k\in[N]}H(\{\mathcal{Z}_{k,j}\}_{j\in[N]\backslash\{k\}})+\frac{N-2T-1}{N-T-1}\sum\limits_{i\in[N]}\sum\limits_{\substack{\mathcal{T}\subseteq [N]\backslash\{i\}\\|\mathcal{T}|=T+1}}\!\!\!H(\{\mathcal{Z}_{i,j}\}_{j\in\mathcal{T}}), \label{converse:N2T+1}
\end{IEEEeqnarray}
where $(a)$ is due to the well-known Han’s inequality \cite[Theorem 17.6.1]{cover1999elements} such that for any $i\in[N]$,
\begin{IEEEeqnarray}{c}
\sum\limits_{\substack{\mathcal{T}\subseteq [N]\backslash\{i\}\\|\mathcal{T}|=T}}H(\{\mathcal{Z}_{i,j}\}_{j\in\mathcal{T}})\geq
\frac{\binom{N-1}{T}}{\binom{N-1}{T+1}}\frac{T}{T+1}\sum\limits_{\substack{\mathcal{T}\subseteq [N]\backslash\{i\}\\|\mathcal{T}|=T+1}}H(\{\mathcal{Z}_{i,j}\}_{j\in\mathcal{T}}). \notag
\end{IEEEeqnarray}

In the case of $N\geq 2T+1$, the inequality in \eqref{converse:N2T+1} can be further derived as follows:
\begin{IEEEeqnarray}{rCl}
\binom{N}{T}(N-T)^2L
&\leq&\binom{N-1}{T}\sum\limits_{k\in[N]}H(\{\mathcal{Z}_{k,j}\}_{j\in[N]\backslash\{k\}})+\frac{N-2T-1}{N-T-1}\sum\limits_{i\in[N]}\sum\limits_{\substack{\mathcal{T}\subseteq [N]\backslash\{i\}\\|\mathcal{T}|=T+1}}\!\!\!H(\{\mathcal{Z}_{i,j}\}_{j\in\mathcal{T}}) \notag\\
&\leq&\binom{N-1}{T}\sum\limits_{k\in[N]}H(\{\mathcal{Z}_{k,j}\}_{j\in[N]\backslash\{k\}})+\frac{N-2T-1}{N-T-1}\binom{N-1}{T+1}\sum\limits_{i\in[N]}H(\{\mathcal{Z}_{i,j}\}_{j\in[N]\backslash\{i\}}) \notag\\
&=&\frac{N-T}{T+1}\binom{N-1}{T}\sum\limits_{i\in[N]}H(\{\mathcal{Z}_{i,j}\}_{j\in[N]\backslash\{i\}}). \notag
\end{IEEEeqnarray}
Then, by \eqref{model:key:distribution} and \eqref{definition:RZ}, we can obtain the following lower bound on the optimal key rate $R_Z^*$:
\begin{IEEEeqnarray}{rCl}
R_Z^* &=&\frac{\sum_{i\in[N]}H(\mathcal{Z}_i)}{H(W_{sum}^{\mathcal{U}})} \notag\\
&\geq&\frac{\sum_{i\in[N]}H(\{\mathcal{Z}_{i,j}\}_{j\in[N]\backslash\{i\}})}{L}\notag\\
&\geq& \frac{\binom{N}{T}(N-T)^2}{\frac{N-T}{T+1}\binom{N-1}{T}} \notag\\
&=&N(T+1). \notag
\end{IEEEeqnarray}

In the case of $T+1< N<2T+1$, the lower bound on $R_Z^*$ can be obtained by transforming the inequality in \eqref{converse:N2T+1} into the following form:
\begin{IEEEeqnarray}{rCl}
\binom{N}{T}(N-T)^2L
&\leq&\binom{N-1}{T}\sum\limits_{k\in[N]}H(\{\mathcal{Z}_{k,j}\}_{j\in[N]\backslash\{k\}})-\frac{2T+1-N}{N-T-1}\sum\limits_{i\in[N]}\sum\limits_{\substack{\mathcal{T}\subseteq [N]\backslash\{i\}\\|\mathcal{T}|=T+1}}\!\!\!H(\{\mathcal{Z}_{i,j}\}_{j\in\mathcal{T}}) \notag\\
&\overset{(a)}{\leq}&\binom{N-1}{T}\sum\limits_{k\in[N]}\!H(\{\mathcal{Z}_{k,j}\}_{j\in[N]\backslash\{k\}})\!-\!\frac{2T+1-N}{N-T-1}\frac{\binom{N-1}{T+1}(T+1)}{N-1}\sum\limits_{i\in[N]}\!H(\{\mathcal{Z}_{i,j}\}_{j\in[N]\backslash\{i\}}) \notag\\
&=&\frac{2(N-T-1)}{N-1}\binom{N-1}{T}\sum\limits_{i\in[N]}H(\{\mathcal{Z}_{i,j}\}_{j\in[N]\backslash\{i\}}) \notag\\
&\overset{(b)}{\leq}&\frac{2(N-T-1)}{N-1}\binom{N-1}{T}L\cdot R_Z^*, \notag
\end{IEEEeqnarray}
where $(a)$ follows from the Han’s inequality \cite[Theorem 17.6.1]{cover1999elements} such that for any $i\in[N]$,
\begin{IEEEeqnarray}{c}
\frac{1}{\binom{N-1}{T+1}}\sum\limits_{\substack{\mathcal{T}\subseteq [N]\backslash\{i\}\\|\mathcal{T}|=T+1}}\frac{H(\{\mathcal{Z}_{i,j}\}_{j\in\mathcal{T}})}{T+1} \geq
\frac{H(\{\mathcal{Z}_{i,j}\}_{j\in[N]\backslash\{i\}})}{N-1}, \notag
\end{IEEEeqnarray}
and $(b)$ follows by \eqref{model:key:distribution} and \eqref{definition:RZ}. Therefore, in the case of $T+1< N<2T+1$, we have
\begin{IEEEeqnarray}{c}
R_Z^*\geq\frac{\binom{N}{T}(N-T)^2}{\frac{2(N-T-1)}{N-1}\binom{N-1}{T}}=\frac{N(N-1)(N-T)}{2(N-T-1)}. \notag
\end{IEEEeqnarray}
Consequently, the lower bound on the optimal key rate $R_Z^*$ is given by
\begin{IEEEeqnarray}{c}
R_Z^*\geq\left\{
\begin{array}{@{}ll}
\frac{N(N-1)(N-T)}{2(N-T-1)},&\text{if}\,\,\, T+1< N<2T+1\\
N(T+1),&\text{if}\,\,\, N\geq 2T+1
\end{array}\right..  \notag
\end{IEEEeqnarray}  

Moreover, in the case of $N\geq 2T+1$, the inequality in \eqref{converse:N2T+1} can also be developed as the following form:
\begin{IEEEeqnarray}{rCl}
\binom{N}{T}(N-T)^2L
&\leq&\binom{N-1}{T}\sum\limits_{k\in[N]}H(\{\mathcal{Z}_{k,j}\}_{j\in[N]\backslash\{k\}})+\frac{N-2T-1}{N-T-1}\sum\limits_{i\in[N]}\sum\limits_{\substack{\mathcal{T}\subseteq [N]\backslash\{i\}\\|\mathcal{T}|=T+1}}\!\!\!H(\{\mathcal{Z}_{i,j}\}_{j\in\mathcal{T}}) \notag\\
&\leq&\binom{N-1}{T}\sum\limits_{k\in[N]}H(\{\mathcal{Z}_{k,j}\}_{j\in[N]\backslash\{k\}})+\frac{N-2T-1}{N-T-1}\sum\limits_{i\in[N]}\sum\limits_{\substack{\mathcal{T}\subseteq [N]\backslash\{i\}\\|\mathcal{T}|=T+1}}\sum\limits_{j\in\mathcal{T}}H(\mathcal{Z}_{i,j}) \notag\\
&=&\binom{N-1}{T}\sum\limits_{k\in[N]}H(\{\mathcal{Z}_{k,j}\}_{j\in[N]\backslash\{k\}})+\frac{N-2T-1}{N-T-1}\binom{N-2}{T}\sum\limits_{i\in[N]}\sum\limits_{j\in[N]\backslash\{i\}}H(\mathcal{Z}_{i,j}). \notag
\end{IEEEeqnarray}
From \eqref{definition:RZ} and \eqref{definition:RK}, we obtain the following lower bound on a linear combination of $R_Z^*$ and $R_K^*$:
\begin{IEEEeqnarray}{c}
N(N-1)(N-T)\leq (N-1)R_Z^*+(N-2T-1)R_K^*. \notag
\end{IEEEeqnarray}

\subsection{Converse for the Masked-Gradient Communication Rate} \label{App:ProofRu}
In this subsection, we establish a lower bound on the optimal masked-gradient communication rate $R_{\mathcal{U}}^*$ for each given dropout set $\mathcal{U}$ with size $0\leq U\leq N-1$ during the masked-gradient upload phase.

For any set of colluding users $\mathcal{T}$ with size $0\leq |\mathcal{T}|<N-1$ and any given user $k\in[N]$,  let $\mathcal{U}\subseteq[N]$ and $\mathcal{S}\subseteq[N]$ denote the dropout sets satisfying $0\leq |\mathcal{U}|<N-1,k\notin\mathcal{U},k\in\mathcal{S}$, and $1\leq |\mathcal{S}| \le N-T-1$. Moreover, let $\mathcal{U}' = \mathcal{U} \cup \{k\}$. Then, the masked gradient $X_k$ sent by user $k\in[N]$ must satisfy
\begin{IEEEeqnarray}{rCl}
H(X_k) & \ge& H(X_k | \{X_i\}_{i \in [N]\setminus\mathcal{U}'}, \{\mathcal{Y}_{i,m}^{\mathcal{U}}\}_{i \in [N]\setminus\mathcal{S},m\in[M_S]}, W_{sum}^{\mathcal{U}'})  \IEEEeqnarraynumspace \notag\\
&\ge& I(X_k ; W_k | \{X_i\}_{i \in [N]\setminus\mathcal{U}'}, \{\mathcal{Y}_{i,m}^{\mathcal{U}}\}_{i \in [N]\setminus\mathcal{S},m\in[M_S]}, W_{sum}^{\mathcal{U}'})  \notag\\
&=& H(W_k | \{X_i\}_{i \in [N]\setminus\mathcal{U}'}, \{\mathcal{Y}_{i,m}^{\mathcal{U}}\}_{i \in [N]\setminus\mathcal{S},m\in[M_S]}, W_{sum}^{\mathcal{U}'}) \nonumber \\
&&-H(W_k | \{X_i\}_{i \in [N]\setminus\mathcal{U}}, \{\mathcal{Y}_{i,m}^{\mathcal{U}}\}_{i \in [N]\setminus\mathcal{S},m\in[M_S]}, W_{sum}^{\mathcal{U}'})  \notag \\
&\overset{(a)}{\geq}& H(W_k | \{W_i\}_{i \in [N]\setminus\mathcal{U}'},\{\mathcal{Z}_i\}_{i\in[N]}, \{W_{i}\}_{i \in [N]\setminus\mathcal{S}},\mathcal{U}, W_{sum}^{\mathcal{U}'}) \nonumber \\
&&-H(W_k | \{X_i\}_{i \in [N]\setminus\mathcal{U}}, \{\mathcal{Y}_{i,m}^{\mathcal{U}}\}_{i \in [N]\setminus\mathcal{S},m\in[M_S]}, W_{sum}^{\mathcal{U}'})   \notag\\
&\overset{(b)}{=}& H(W_k | \{W_i\}_{i \in [N]\setminus\mathcal{U}'},\{\mathcal{Z}_i\}_{i\in[N]}, \{W_{i}\}_{i \in [N]\setminus\mathcal{S}}, W_{sum}^{\mathcal{U}'}) \nonumber \\
&&-H(W_k | \{X_i\}_{i \in [N]\setminus\mathcal{U}}, \{\mathcal{Y}_{i,m}^{\mathcal{U}}\}_{i \in [N]\setminus\mathcal{S},m\in[M_S]}, W_{sum}^{\mathcal{U}'})   \notag\\
&\overset{(c)}{=}& L -H(W_k | \{X_i\}_{i \in [N]\setminus\mathcal{U}}, \{\mathcal{Y}_{i,m}^{\mathcal{U}}\}_{i \in [N]\setminus\mathcal{S},m\in[M_S]}, W_{sum}^{\mathcal{U}'})  \notag\\
&\overset{(d)}{=}& L -H(W_k | \{X_i\}_{i \in [N]\setminus\mathcal{U}},\{\mathcal{Y}_{i,m}^{\mathcal{U}}\}_{i \in [N]\setminus\mathcal{S},m\in[M_S]},W_{sum}^{\mathcal{U}}, W_{sum}^{\mathcal{U}'}) \notag\\
&\overset{(e)}{=}& L, \label{eq:ru_final_bound}
\end{IEEEeqnarray}
where $(a)$ follows from the facts that $X_i$ is the function of $W_i$ and $\{\mathcal{Z}_j\}_{j\in[N]}$ for any $i\in[N]\backslash\mathcal{U}'$ by \eqref{model:key:distribution}-\eqref{model:masked:graidient:2} and that $\{\mathcal{Y}_{i,m}^{\mathcal{U}}\}_{m\in[M_S]}$ is the  function of $\{W_i,\{\mathcal{Z}_j\}_{j\in[N]},\mathcal{U}\}$ for any $i \in [N]\backslash\mathcal{S}$ by \eqref{model:key:distribution} and \eqref{model:masked:gradients};
$(b)$ follows by the fact that the dropout set $\mathcal{U}$ is publicly known and is treated as a fixed subset of users;
$(c)$ follows from the fact that the random variables $\{\{W_i\}_{i\in[N]},\{\mathcal{Z}_i\}_{i\in[N]}\}$ are generated independently of each other by \eqref{model:key:inden} and therefore $W_k$ is independent of $\{W_i\}_{i \in [N]\setminus\mathcal{U}'},\{\mathcal{Z}_i\}_{i\in[N]}, \{W_{i}\}_{i \in [N]\setminus\mathcal{S}}$, and $W_{sum}^{\mathcal{U}'}$;
$(d)$ is due to the correctness constraint in \eqref{eq:correctness};
$(e)$ is because $W_{sum}^{\mathcal{U}}-W_{sum}^{\mathcal{U}'}=W_k$.
This implies that, in order to ensure that the server can correctly recover the desired gradient sum $W_{sum}^{\mathcal{U}}$ for every feasible dropout set $\mathcal{U}\subseteq[N]$ with $0\leq|\mathcal{U}|\leq N-1$, the masked gradient $X_k$
transmitted by each user $k\in[N]$ must have entropy at least $L$.
 
Therefore, for any set of dropout users $\mathcal{U}\subseteq[N]$ with $|\mathcal{U}|=U$ and $0\leq U\leq N-1$, by \eqref{performance:RU} and \eqref{eq:ru_final_bound}, we obtain the following lower bound on the optimal masked-gradient communication rate $R_{\mathcal{U}}^*$:
\begin{IEEEeqnarray}{c}
R_{\mathcal{U}}^*=\frac{\sum\limits_{k\in [N]\backslash \mathcal{U}}H(X_k)}{L}\geq N-U.  \notag
\end{IEEEeqnarray}

\subsection{Converse for the Mask-Elimination Communication Rate}\label{App:ProofRs}

In this subsection, for any given dropout set $\mathcal{U}\subseteq[N]$ with size $0\leq U\leq N-1$ in the masked-gradient upload phase, we establish a lower bound on the optimal masked-elimination communication rate $R_S^{\mathcal{U}*}$ for $1\leq S\leq N-T-1$ and for $S=0$ with $T<N-U$. The case $S=0$ with $T\geq N-U$ remains open.

In the case $1\leq S\leq N-T-1$, for each such dropout set $\mathcal{U}\subseteq[N]$ with size $0\leq U\leq N-1$, let $\mathcal{S}\subseteq[N]$ be the dropout set in the mask-elimination phase such that $\mathcal{S}\cap([N]\backslash\mathcal{U})\neq\emptyset$. Then, for any set of colluding users $\mathcal{T} \subseteq [N]\backslash\mathcal{S}$ with  $|\mathcal{T}|=T$, we must have $[N]\backslash\mathcal{T}\nsubseteq\mathcal{U}$. Therefore, we obtain
\begin{IEEEeqnarray}{rCl}
L
&\overset{(a)}{=}&I(W_{sum}^{\mathcal{U}}; \{X_i\}_{i \in [N]\backslash\mathcal{U} }, \{\mathcal{Y}_{i,m}^{\mathcal{U}}\}_{i \in[N]\backslash\mathcal{S} ,m\in[M_S]}) \notag \\
&=&I(W_{sum}^{\mathcal{U}}; \{X_i\}_{i \in [N]\backslash\mathcal{U} }, \{\mathcal{Y}_{i,m}^{\mathcal{U}}\}_{i \in\mathcal{T} ,m\in[M_S]},\{\mathcal{Y}_{i,m}^{\mathcal{U}}\}_{i \in[N]\backslash(\mathcal{S}\cup\mathcal{T}) ,m\in[M_S]}) \notag \\
&=& I(W_{sum}^{\mathcal{U}}; \{X_i\}_{i \in [N]\backslash\mathcal{U} }, \{\mathcal{Y}_{i,m}^{\mathcal{U}}\}_{i \in\mathcal{T} ,m\in[M_S]}) \notag\\
&&+I(W_{sum}^{\mathcal{U}}; \{\mathcal{Y}_{i,m}^{\mathcal{U}}\}_{i \in[N]\backslash(\mathcal{S}\cup\mathcal{T}), m\in[M_S]}|\{X_i\}_{i \in [N]\backslash\mathcal{U} }, \{\mathcal{Y}_{i,m}^{\mathcal{U}}\}_{i \in\mathcal{T} ,m\in[M_S]}) \notag \\
&\overset{(b)}{\leq}&I(W_{sum}^{\mathcal{U}}; \{X_i\}_{i \in [N]\backslash\mathcal{U} }, \{W_i\}_{i\in\mathcal{T}},\{\mathcal{Z}_i\}_{i\in\mathcal{T}},\{\mathcal{D}_{i}\}_{i\in\mathcal{T}},\mathcal{U}) \notag\\
&&+I(W_{sum}^{\mathcal{U}}; \{\mathcal{Y}_{i,m}^{\mathcal{U}}\}_{i \in[N]\backslash(\mathcal{S}\cup\mathcal{T}), m\in[M_S]}|\{X_i\}_{i \in [N]\backslash\mathcal{U} }, \{\mathcal{Y}_{i,m}^{\mathcal{U}}\}_{i \in\mathcal{T} ,m\in[M_S]}) \notag\\
&\overset{(c)}{=}&I(W_{sum}^{\mathcal{U}}; \{\mathcal{Y}_{i,m}^{\mathcal{U}}\}_{i \in[N]\backslash(\mathcal{S}\cup\mathcal{T}), m\in[M_S]}|\{X_i\}_{i \in [N]\backslash\mathcal{U} }, \{\mathcal{Y}_{i,m}^{\mathcal{U}}\}_{i \in\mathcal{T} ,m\in[M_S]}) \notag \\
&\leq&H(\{\mathcal{Y}_{i,m}^{\mathcal{U}}\}_{i \in[N]\backslash(\mathcal{S}\cup\mathcal{T}), m\in[M_S]}|\{X_i\}_{i \in [N]\backslash\mathcal{U} }, \{\mathcal{Y}_{i,m}^{\mathcal{U}}\}_{i \in\mathcal{T} ,m\in[M_S]}) \notag \\
&\leq&H(\{\mathcal{Y}_{i,m}^{\mathcal{U}}\}_{i \in[N]\backslash(\mathcal{S}\cup\mathcal{T}), m\in[M_S]}), \label{converse:3}
\end{IEEEeqnarray}
where $(a)$ follows from the correctness constraint in \eqref{eq:correctness};
$(b)$ follows from the fact that $\{\mathcal{Y}_{i,m}^{\mathcal{U}}\}_{m\in[M_S]}$ is the deterministic function of $\{W_i,\mathcal{Z}_i, \mathcal{D}_{i},\mathcal{U}\}$ for any $i \in \mathcal{T}$ by \eqref{model:masked:gradients} and \eqref{def:Di};
$(c)$ follows by Lemma \ref{Lemma:converse} and the fact that the dropout set $\mathcal{U}$ is publicly known and is treated as a fixed subset of users. 
Furthermore, 
by summing the inequality in \eqref{converse:3} over all subsets $\mathcal{T} \subseteq [N]\backslash\mathcal{S}$ with size $|\mathcal{T}|=T$, 
\begin{IEEEeqnarray}{rCl}
\binom{N-S}{T} L &\leq& \sum\limits_{\mathcal{T} \subseteq [N]\backslash\mathcal{S},|\mathcal{T}|=T} H(\{\mathcal{Y}_{i,m}^{\mathcal{U}}\}_{i \in[N]\backslash(\mathcal{S}\cup\mathcal{T}), m\in[M_S]}) \notag \\
&\leq&\sum\limits_{\mathcal{T} \subseteq [N]\backslash\mathcal{S},|\mathcal{T}|=T}\sum\limits_{i \in[N]\backslash(\mathcal{S}\cup\mathcal{T})} H(\{\mathcal{Y}_{i,m}^{\mathcal{U}}\}_{m\in[M_S]}) \notag \\
&=&\binom{N-S-1}{T}\sum\limits_{i \in[N]\backslash\mathcal{S}} H(\{\mathcal{Y}_{i,m}^{\mathcal{U}}\}_{m\in[M_S]}). \label{converse:4} \IEEEeqnarraynumspace
\end{IEEEeqnarray}
Therefore, for any given dropout set $\mathcal{S}\subseteq[N]$ satisfying $\mathcal{S}\cap([N]\backslash\mathcal{U})\neq\emptyset$  and $1\le S\le N-T-1$, by \eqref{performance:RS} and \eqref{converse:4}, we obtain the following lower bound on the optimal masked-elimination communication rate $R_S^{\mathcal{U}*}$:
\begin{IEEEeqnarray}{rCl}
R_S^{\mathcal{U}*}&=&\frac{\max\limits_{\mathcal{S}'\subseteq[N],|\mathcal{S}'|=S}\sum\limits_{i\in[N]\backslash\mathcal{S}'}\sum\limits_{m\in[M_{S}]}H(\mathcal{Y}_{i,m}^{\mathcal{U}})}{H(W_{sum}^{\mathcal{U}})} \notag \\
&\geq&\frac{\sum\limits_{i\in[N]\backslash\mathcal{S}}\sum\limits_{m\in[M_S]}H(\mathcal{Y}_{i,m}^{\mathcal{U}})}{L} \notag\\
&\geq&\frac{\sum\limits_{i\in[N]\backslash\mathcal{S}}H(\{\mathcal{Y}_{i,m}^{\mathcal{U}}\}_{m\in[M_S]})}{L} \notag\\
&\geq&\frac{\binom{N-S}{T}}{\binom{N-S-1}{T}} \notag\\
&=&\frac{N-S}{N-S-T}. \notag
\end{IEEEeqnarray}

In the case $S=0$ with $T<N-U$, we have $\mathcal{S}=\emptyset$. Then, for each given dropout set $\mathcal{U}\subseteq[N]$ with size $0\leq U\leq N-1$ and any set of colluding users $\mathcal{T} \subseteq [N]\backslash\mathcal{S}$ with $|\mathcal{T}|=T$, we must have $[N]\backslash\mathcal{T}\nsubseteq\mathcal{U}$. Therefore, the above proof remains fully valid in this case, and accordingly, we obtain 
\begin{IEEEeqnarray}{rCl}
R_S^{\mathcal{U}*}\geq\frac{N-S}{N-S-T}. \notag
\end{IEEEeqnarray}

\section{Comparison with Related Work}\label{section:Comparison}
The works most closely related to this paper are \cite{zhao2022information,so2022lightsecagg,zhang2025securegroup}, which consider the $T$-colluding information-theoretic secure aggregation problem with user dropouts.
Although their formulations impose a set-inclusion relation between the dropout sets in the masked-gradient upload and mask-elimination phases, their proposed schemes can also directly support the dropout scenarios considered in this paper. This is because mask-elimination messages in these schemes are generated from the established correlated keys, allowing a user to participate in the mask-elimination phase regardless of whether it successfully uploaded its masked gradient in the preceding phase.
Table~\ref{tab:comparison} compares the performance of these schemes with that of the proposed scheme. 
For a fair comparison, we consider the setting in which all schemes are required to support all feasible dropout scenarios satisfying $0\leq U\leq N-1$ and $0\leq S\leq N-T-1$.
The key-distribution communication rate is not reported for the schemes in \cite{zhao2022information,zhang2025securegroup}. The scheme in \cite{zhao2022information} assumes that the correlated random keys are established by a trusted third party, whereas the scheme in \cite{zhang2025securegroup} assumes a prescribed symmetric groupwise key-distribution structure. Consequently, the communication cost required to establish the correlated randomness is not explicitly modeled in either scheme.

Compared with the schemes in \cite{zhao2022information,so2022lightsecagg,zhang2025securegroup}, the main advantage of the proposed scheme is its ability to reduce the mask-elimination communication rate by adapting to the realized number of dropouts. To tolerate every feasible number of dropouts $0\leq S\leq N-T-1$ during the mask-elimination phase, the previous
schemes are designed for the worst-case value $S=N-T-1$ and therefore incur a fixed mask-elimination communication rate of $T+1$, regardless of the realized value of $S$. By contrast, the proposed scheme dynamically adjusts its communication cost to the realized number of dropouts and achieves the mask-elimination communication rate $\frac{N-S}{N-T-S}$ for every feasible $S$. This rate is no larger than $T+1$, with equality in the worst-case scenario $S=N-T-1$. For $T>0$, it is strictly smaller
whenever $S<N-T-1$. 
It is worth emphasizing that this communication gain arises from adapting the mask-elimination communication to the realized dropout scenario, rather than from relaxing the set-inclusion constraint between the dropout sets in the masked-gradient upload and mask-elimination phases.

Another important distinction between this paper and the existing works in \cite{zhao2022information,so2022lightsecagg,zhang2025securegroup}
lies in the fundamental-limit characterization. These works do not investigate the fundamental limits of either the key rate or the key-distribution communication rate. In contrast, we establish, for the first time, information-theoretic lower bounds on these two rates and prove that the rates achieved by the proposed scheme are either optimal or within a multiplicative factor of two of their respective information-theoretic optima.
Moreover, existing works characterize the mask-elimination communication rate by taking the worst case over the dropout scenarios in both the masked-gradient upload and mask-elimination phases. Such a joint worst-case characterization only establishes  the optimality of a single worst-case mask-elimination communication rate. It does not reveal how the required mask-elimination
communication varies with the particular dropout set realized during the masked-gradient upload phase.
In contrast, for each given dropout set realized during the masked-gradient upload phase, we characterize the fundamental limits of the masked-gradient communication rate and the mask-elimination communication rate. The resulting converse bounds establish the optimality of the proposed scheme in terms of the masked-gradient communication rate for all feasible parameters and in terms of the mask-elimination communication rate whenever $1\leq S\leq N-T-1$ or
$(S=0 \text{ and } T<N-U)$. Together with the converse bounds on the
key rate and key-distribution communication rate, these results provide a unified fundamental-limit characterization of all four rates considered in this paper.


\begin{table}[htbp]
\centering
\caption{Performance comparison between our proposed scheme and the secure aggregation schemes in \cite{zhao2022information,so2022lightsecagg,zhang2025securegroup}.}
\label{tab:comparison}
\begin{NiceTabular}{@{\;}l@{\;}c@{\;}c@{\;}c@{\;}c@{\;}}[hvlines]
\rule[-3.0mm]{0pt}{8mm}\diagbox{\raisebox{0.5mm}{Performance Metrics}}{Secure Aggregation Schemes}
&  Zhao-Sun Scheme \cite{zhao2022information} & So et al.'s Scheme \cite{so2022lightsecagg} & Zhang et al.'s Scheme \cite{zhang2025securegroup} & Our Scheme \\
Key Rate $R_Z$ & $N+T\sum_{g=T+1}^{N}\binom{N}{g}$ & $N(T+1)$ & $\binom{N}{N-T}(N-T)$ & $N(T+1)$ \\ 
Key-Distribution Communication Rate $R_K$ & \cmarkk & $N(N-1)$ & \cmarkk & $N(N-1)$ \\
Masked-Gradient Communication Rate $R_{\mathcal{U}}$ & $N-U$ & $N-U$ & $N-U$ & $N-U$ \\ 
Mask-Elimination Communication Rate $R_S^{\mathcal{U}}$ & $T+1$ & $T+1$ & $T+1$ & $\frac{N-S}{N-T-S}$ \\
Required Finite Field $\mathbb{F}_q$ & $q\geq N+T+1$ & $q\geq N$ & $q\rightarrow\infty$ & $q\geq N$ \\
Key Establishment Mechanism & Trusted Third Party & User Collaboration & Groupwise Key & User Collaboration \\
Construction Type & Explicit & Explicit & Probabilistic & Explicit \end{NiceTabular}
\end{table}

\section{Conclusion}\label{conclusion}
In this paper, we investigated secure aggregation with user dropouts under the $T$-colluding security constraint. We formulated an adaptive secure aggregation framework that  supports all feasible dropout scenarios, thereby enabling the communication strategy to dynamically adapt to the actual dropout scenario observed during protocol execution. Within this framework, we established converse bounds on four fundamental resources: the randomness required to generate correlated keys, the communication required to distribute them, the communication for uploading masked gradients, and the additional communication required to eliminate residual masks. In particular, the converse bounds on the latter two resources are characterized separately for each dropout set realized during the masked-gradient upload phase.
For a system with $N$ users, we developed an adaptive secure aggregation scheme that achieves the rate tuple
$\big(N(T+1), N(N-1), N-U, \frac{N-S}{N-T-S} \big)$
for all feasible dropout scenarios, where $U$ and $S$ denote the actual numbers of dropout users during protocol execution in the masked-gradient upload and mask-elimination phases, respectively. By comparing these achievable rates with the converse bounds, we showed that the proposed scheme achieves the optimal masked-gradient communication rate for every feasible dropout set. It also achieves the optimal mask-elimination communication rate whenever $1\leq S\leq N-T-1$ or $(S=0 \text{ and } T<N-U)$. Moreover, the achieved key rate and key-distribution communication rate are each either information-theoretically optimal or within a multiplicative factor of two of their respective optima. 
Overall, this paper provides the first fundamental characterization of the communication and randomness resources required for secure aggregation with user dropouts. 
Future work includes determining the optimal mask-elimination communication rate in the remaining regime where $S=0$ and $T\geq N-U$, closing the remaining gaps in the key rate and the key-distribution communication rate, and extending the proposed framework to incorporate more practical system constraints.

\appendix[Security Proof of the Proposed Scheme]\label{appendix:security}

In this appendix, we formally prove that the proposed secure aggregation scheme satisfies the privacy constraint in \eqref{eq:security}.
For any colluding user set $\mathcal{T}\subseteq[N]$ with $|\mathcal{T}|\leq T$, let $\widetilde{\mathcal{T}}\subseteq[N]$ be an arbitrary set satisfying $|\widetilde{\mathcal{T}}|=T$ and $\mathcal{T}\subseteq\widetilde{\mathcal{T}}$. Then, for any given dropout scenario satisfying $\mathcal{U}\subseteq[N]$ with size $0\leq U\leq N-1$ and $\mathcal{S}\subseteq[N]$ with size $0\leq S\leq N-T-1$, the proposed secure aggregation scheme satisfies 
\begin{IEEEeqnarray}{rCl}
&& I\Big(\{{W}_i\}_{i \in [N]} ; \{\mathcal{Z}_i\}_{i\in\mathcal{T}},\{\mathcal{Z}_{i,j},\mathcal{Z}_{j,i}\}_{j\in[N]\backslash\{i\},i\in\mathcal{T}},\{X_i\}_{i\in[N]},\{\mathcal{Y}_{i,m}^{\mathcal{U}}\}_{i\in[N],m\in[M]} \Big| W_{sum}^{\mathcal{U}}, \{W_i\}_{i\in\mathcal{T}} \Big)  \notag\\
&\overset{(a)}{\leq}& I\Big(\{{W}_i\}_{i \in [N]} ; \{\mathcal{Z}_i\}_{i\in\mathcal{T}},\{\mathcal{Z}_{j,i}\}_{j\in[N],i\in\widetilde{\mathcal{T}}},\{X_i\}_{i\in[N]},\{\mathcal{Y}_{i,m}^{\mathcal{U}}\}_{i\in[N],m\in[M]} \Big| W_{sum}^{\mathcal{U}}, \{W_i\}_{i\in\mathcal{T}} \Big)  \notag\\
&\overset{(b)}{=}&I\Big(\{{W}_i\}_{i \in [N]} ; \{C_i\}_{i\in\mathcal{T}}, \{D^{(i,m)}_{k,\ell}\}_{i\in\mathcal{T},m\in[M],k\in[T],\ell\in[\lambda_m]},
\{f^{(j,m)}_{\ell}(\alpha_i)\}_{j\in[N],m\in[M],\ell\in[\lambda_m],i\in\widetilde{\mathcal{T}}},\notag\\
&&\quad\qquad\qquad\qquad\qquad\qquad\qquad\quad \{W_i+C_i\}_{i\in[N]},\{f^{(\mathcal{U},m)}_{\ell}(\alpha_i)\}_{m\in[M],\ell\in[\lambda_m],i\in[N]} \Big| W_{sum}^{\mathcal{U}}, \{W_i\}_{i\in\mathcal{T}} \Big) \notag\\
&\overset{(c)}{=}&I\Big(\{{W}_i\}_{i \in [N]} ; \{C_i\}_{i\in\mathcal{T}}, \{D^{(i,m)}_{k,\ell}\}_{i\in\mathcal{T},m\in[M],k\in[T],\ell\in[\lambda_m]},
\{f^{(j,m)}_{\ell}(\alpha_i)\}_{j\in[N],m\in[M],\ell\in[\lambda_m],i\in\widetilde{\mathcal{T}}},\sum_{j\in[N]\backslash\mathcal{U}}C_j,\notag\\
&&\quad \{\sum_{j\in[N]\backslash\mathcal{U}}f^{(j,m)}_{\ell}(\alpha_i)\}_{m\in[M],\ell\in[\lambda_m],i\in\widetilde{\mathcal{T}}}, \{W_i+C_i\}_{i\in[N]},\{f^{(\mathcal{U},m)}_{\ell}(\alpha_i)\}_{m\in[M],\ell\in[\lambda_m],i\in[N]} \Big| W_{sum}^{\mathcal{U}}, \{W_i\}_{i\in\mathcal{T}} \Big) \notag\\
&\overset{(d)}{=}&I\Big(\{{W}_i\}_{i \in [N]} ; \{C_i\}_{i\in\mathcal{T}}, \{D^{(i,m)}_{k,\ell}\}_{i\in\mathcal{T},m\in[M],k\in[T],\ell\in[\lambda_m]},
\{f^{(j,m)}_{\ell}(\alpha_i)\}_{j\in[N],m\in[M],\ell\in[\lambda_m],i\in\widetilde{\mathcal{T}}},\sum_{j\in[N]\backslash\mathcal{U}}C_j,\notag\\
&&\quad \{\sum_{j\in[N]\backslash\mathcal{U}}P^{(j,m)}_{k,\ell}\}_{m\in[M],k\in[N-T-m+1],\ell\in[\lambda_m]}, \{\sum_{j\in[N]\backslash\mathcal{U}}D^{(j,m)}_{k,\ell}\}_{m\in[M],k\in[T],\ell\in[\lambda_m]}, \notag\\
&&\qquad \{\sum_{j\in[N]\backslash\mathcal{U}}\!\!f^{(j,m)}_{\ell}(\alpha_i)\}_{m\in[M],\ell\in[\lambda_m],i\in\widetilde{\mathcal{T}}}, \{W_i\!+\!C_i\}_{i\in[N]},\{f^{(\mathcal{U},m)}_{\ell}(\alpha_i)\}_{m\in[M],\ell\in[\lambda_m],i\in[N]} \Big| W_{sum}^{\mathcal{U}}, \{W_i\}_{i\in\mathcal{T}} \Big) \notag\\
&\overset{(e)}{=}&I\Big(\{{W}_i\}_{i \in [N]} ; \{C_i\}_{i\in\mathcal{T}}, \{D^{(i,m)}_{k,\ell}\}_{i\in\mathcal{T},m\in[M],k\in[T],\ell\in[\lambda_m]},
\{f^{(j,m)}_{\ell}(\alpha_i)\}_{j\in[N],m\in[M],\ell\in[\lambda_m],i\in\widetilde{\mathcal{T}}},\sum_{j\in[N]\backslash\mathcal{U}}C_j,\notag\\
&&\quad\qquad \{\sum_{j\in[N]\backslash\mathcal{U}}P^{(j,m)}_{k,\ell}\}_{m\in[M],k\in[N-T-m+1],\ell\in[\lambda_m]}, \{\sum_{j\in[N]\backslash\mathcal{U}}D^{(j,m)}_{k,\ell}\}_{m\in[M],k\in[T],\ell\in[\lambda_m]}, \notag\\
&&\qquad\qquad\qquad\qquad\qquad \{\sum_{j\in[N]\backslash\mathcal{U}}f^{(j,m)}_{\ell}(\alpha_i)\}_{m\in[M],\ell\in[\lambda_m],i\in\widetilde{\mathcal{T}}}, \{W_i+C_i\}_{i\in[N]} \Big| W_{sum}^{\mathcal{U}}, \{W_i\}_{i\in\mathcal{T}} \Big) \notag\\
&\overset{(f)}{=}&I\Big(\{{W}_i\}_{i \in [N]} ; \{C_i\}_{i\in\mathcal{T}}, \{D^{(i,m)}_{k,\ell}\}_{i\in\mathcal{T},m\in[M],k\in[T],\ell\in[\lambda_m]},
\{f^{(j,m)}_{\ell}(\alpha_i)\}_{j\in[N],m\in[M],\ell\in[\lambda_m],i\in\widetilde{\mathcal{T}}},\notag\\
&&\quad\qquad\qquad\qquad\qquad\qquad\qquad\qquad\qquad\qquad\qquad\qquad\qquad\qquad \{W_i+C_i\}_{i\in[N]} \Big| W_{sum}^{\mathcal{U}}, \{W_i\}_{i\in\mathcal{T}} \Big) \notag\\
&\overset{(g)}{=}&I\Big(\{{W}_i\}_{i \in [N]} ; \{C_i\}_{i\in\mathcal{T}}, \{D^{(i,m)}_{k,\ell}\}_{i\in\mathcal{T},m\in[M],k\in[T],\ell\in[\lambda_m]},
\{f^{(j,m)}_{\ell}(\alpha_i)\}_{j\in[N]\backslash\mathcal{T},m\in[M],\ell\in[\lambda_m],i\in\widetilde{\mathcal{T}}},\notag\\
&&\qquad\qquad\qquad\qquad\qquad\qquad\qquad\qquad\qquad\qquad\qquad\quad\qquad\qquad \{W_i+C_i\}_{i\in[N]\backslash\mathcal{T}} \Big| W_{sum}^{\mathcal{U}}, \{W_i\}_{i\in\mathcal{T}} \Big) \label{security:proof} \\
&\overset{(h)}{=}&I\Big(\{{W}_i\}_{i \in [N]} ; 
\{f^{(j,m)}_{\ell}(\alpha_i)\}_{j\in[N]\backslash\mathcal{T},m\in[M],\ell\in[\lambda_m],i\in\widetilde{\mathcal{T}}}, \{W_i+C_i\}_{i\in[N]\backslash\mathcal{T}} \Big| W_{sum}^{\mathcal{U}}, \{W_i\}_{i\in\mathcal{T}} \Big) \notag\\
&\leq&I\Big(\{{W}_i\}_{i \in [N]} ; 
\{f^{(j,m)}_{\ell}(\alpha_i)\}_{j\in[N]\backslash\mathcal{T},m\in[M],\ell\in[\lambda_m],i\in\widetilde{\mathcal{T}}}, \{W_i+C_i\}_{i\in[N]} \Big) \notag\\
&=&I\Big(\{{W}_i\}_{i \in [N]} ; 
\{f^{(j,m)}_{\ell}(\alpha_i)\}_{j\in[N]\backslash\mathcal{T},m\in[M],\ell\in[\lambda_m],i\in\widetilde{\mathcal{T}}} \Big) \notag\\
&&\qquad\qquad\qquad\qquad\qquad\qquad\qquad + I\Big(\{{W}_i\}_{i \in [N]} ; \{W_i+C_i\}_{i\in[N]}\Big|\{f^{(j,m)}_{\ell}(\alpha_i)\}_{j\in[N]\backslash\mathcal{T},m\in[M],\ell\in[\lambda_m],i\in\widetilde{\mathcal{T}}} \Big) \notag\\
&\overset{(i)}{=}&I\Big(\{{W}_i\}_{i \in [N]} ; \{W_i+C_i\}_{i\in[N]}\Big|\{f^{(j,m)}_{\ell}(\alpha_i)\}_{j\in[N]\backslash\mathcal{T},m\in[M],\ell\in[\lambda_m],i\in\widetilde{\mathcal{T}}} \Big) \notag\\
&\overset{(j)}{=}&I\Big(\{{W}_i\}_{i \in [N]} ; \{W_i+C_i\}_{i\in[N]} \Big) \notag\\
&\overset{(k)}{=}&0, \notag
\end{IEEEeqnarray}
where $(a)$ follows from the fact that $\{\mathcal{Z}_{i,j}\}_{j\in[N]\backslash\{i\}}$ is a deterministic function of $\mathcal{Z}_i$ for any $i\in\mathcal{T}$ by \eqref{encoding:polynomial}--\eqref{encoded:key};
$(b)$ is due to \eqref{scheme;random:key}--\eqref{uplaod:X} and \eqref{encoding:message};
$(c)$ follows from the fact that $\sum_{j\in[N]\backslash\mathcal{U}}C_j$ can be recovered from $\{W_i+C_i\}_{i\in[N]\backslash\mathcal{U}}$ and $W_{sum}^{\mathcal{U}}$;
$(d)$ follows from the facts that $\{\sum_{j\in[N]\backslash\mathcal{U}}P^{(j,m)}_{k,\ell}\}_{m\in[M],k\in[N-T-m+1],\ell\in[\lambda_m]}$ is a function of $\sum_{j\in[N]\backslash\mathcal{U}}C_j$ by \eqref{intical:case:P}-\eqref{matrix:P}, and that $\{\sum_{j\in[N]\backslash\mathcal{U}}D^{(j,m)}_{k,\ell}\}_{k\in[T]}$ is a deterministic function of $\{\sum_{j\in[N]\backslash\mathcal{U}}P^{(j,m)}_{k,\ell}\}_{k\in[N-T-m+1]}$ and $\{\sum_{j\in[N]\backslash\mathcal{U}}f^{(j,m)}_{\ell}(\alpha_i)\}_{i\in\widetilde{\mathcal{T}}}$ for any given $m\in[M]$ and $\ell\in[\lambda_m]$ by \eqref{encoding:polynomial};
$(e)$ follows from \eqref{encoding:polynomial:sum};
$(f)$ follows from arguments similar to those used in $(c)$ and $(d)$;
$(g)$ follows from the fact that $\{f^{(j,m)}_{\ell}(\alpha_i)\}_{i\in\widetilde{\mathcal{T}}}$ is a deterministic function of $C_j$ and $\{D^{(j,m)}_{k,\ell}\}_{k\in[T]}$ for any given $j\in\mathcal{T},m\in[M]$, and $\ell\in[\lambda_m]$ by \eqref{intical:case:P}--\eqref{encoding:polynomial};
$(h)$ follows from the fact that $\{C_i\}_{i\in\mathcal{T}}$ and $\{D^{(i,m)}_{k,\ell}\}_{i\in\mathcal{T},m\in[M],k\in[T],\ell\in[\lambda_m]}$ are generated independently of all other random variables appearing in \eqref{security:proof};
$(i)$ follows from the fact that $\{f^{(j,m)}_{\ell}(\alpha_i)\}_{j\in[N]\backslash\mathcal{T},m\in[M],\ell\in[\lambda_m],i\in\widetilde{\mathcal{T}}}$ is generated independently of the gradients $\{{W}_i\}_{i \in [N]}$ by \eqref{encoding:polynomial}; 
$(j)$ follows from the fact that the evaluations
$\{f^{(j,m)}_{\ell}(\alpha_i)\}_{i\in\widetilde{\mathcal{T}}}$ are generated using the $T$ random noise symbols $\{D^{(j,m)}_{k,\ell}\}_{k\in[T]}$ for any
$j\in[N]\backslash\mathcal{T},m\in[M]$, and $\ell\in[\lambda_m]$ by \eqref{encoding:polynomial}, while $\{D^{(j,m)}_{k,\ell}\}_{j\in[N]\backslash\mathcal{T},m\in[M],k\in[T],\ell\in[\lambda_m]}$ is independent of $\{W_i,C_i\}_{i\in[N]}$;
$(k)$ follows from the fact that $C_1,C_2,\ldots,C_N$ are mutually independent and uniformly distributed over the finite field $\mathbb{F}_q$.

\bibliographystyle{ieeetr}
\bibliography{reference.bib}

\end{document}